\documentclass{article}
\usepackage[ruled, vlined]{algorithm2e}
\usepackage{amsmath, amssymb, amsthm}
\usepackage{xcolor}
\usepackage{bbm}
\usepackage{tcolorbox}
\usepackage{tikz}
\usepackage{float}

\providecommand{\keywords}[1]
{
  \small	
  \textbf{\textit{Keywords---}} #1
}

\providecommand{\declarations}[1]
{
  \small	
  \textbf{\textit{Declarations---}} #1
}

\providecommand{\coi}[1]
{
  \small	
  \textbf{\textit{Conflict of Interest---}} #1
}

\begin{document}

\newtheorem{thm}{Theorem}
\newtheorem{rslt}{Result}

\begin{center}
\begin{large}
{\bf A Novel Metric for  Detecting Anomalous Ship Behavior Using a Variation of the DBSCAN Clustering Algorithm}
\end{large} 

\vspace{.5cm}

Carsten H. Botts

Carsten.Botts@jhuapl.edu

The Johns Hopkins University Applied Physics Lab

Laurel, MD

\end{center}

\abstract{There is a growing need to quickly and accurately identify anomalous behavior in ships.  This paper applies a variation of the  Density Based Spatial Clustering Among Noise (DBSCAN)  algorithm to identify such anomalous behavior given a ship's Automatic Identification System (AIS) data. This variation of the DBSCAN algorithm has been previously introduced in the literature, and in this study, we elucidate and explore the mathematical details of this algorithm, we  introduce a novel anomaly metric which is more statistically informative than the one previously suggested, and we study the asymptotic properties of this metric.}

\vspace{.5cm}


\vspace{.5cm}

\declarations{Funding: Not applicable. Conflicts of Interest/Competing Interest: None. Availability of data and material: data used is available on {\tt MarineCadastre.gov}. Code availability: Not applicable.}

\vspace{.5cm}

\coi{On behalf of all authors, the corresponding author states that there is no conflict of interest.} 

\section{Introduction}

In this paper we identify anomalous behavior in ships given their Automatic Identification System (AIS) data.   AIS data is reported by all ships, and among other things, these data include a ship's position (latitude \& longitude), speed, and course (direction) over time. We use this data to  identify spatial and behavioral patterns of ships. With these patterns identified from the training data, we  can potentially spot anomalous behavior in other ships given their new/incoming AIS data.   We do this by  applying a variation to the DBSCAN (Density Based Spatial Clustering Among Noise) algorithm (\cite{Ester}). The DBSCAN algorithm is  used to create spatial clusters of data (see
 \cite{Erman}, \cite{Schlitt}, \cite{Yu}), and because of its simplicity and speed, it is becoming an increasingly popular method used to identify anomalous behavior in ships (see \cite{Arguedas}, \cite{Pallotta2013a}, \cite{Pallotta2013b}, \cite{Riveiro}).   The variation to the DBSCAN algorithm we consider was  introduced by Liu et al. (\cite{Liu1}). In this paper, we review, correct, and expand on  the anomaly detection methods given in Liu (\cite{Liu1} \& \cite{Liu2}), and we also offer an alternative and more statistically informative way of identifying anomalous behavior.

This paper begins by showing the data and briefly discussing the problem at hand. This is done in Section \ref{sctn:Data}. Simply stated, the problem is this: ``How does one identify anomalous behavior among the ships that are reporting AIS data?" An easy way to answer such a question would be to create spatial clusters of AIS positional data using the DBSCAN algorithm.  If a new ship then reported AIS positional data which was geographically far from any of these clusters, such data may be considered unusual, or anomalous. Methods such as these have been applied in a variety of maritime settings. Lee et al. (\cite{Lee}) applied it to ships in port, and Wang et al. (\cite{Wang}) applied it to ships in waterways. There is so much data to create these positional/spatial clusters that Chen et al. (\cite{Chen}) studied how to form such clusters on the Hadoop platform.

 But what if one wished to detect anomalies with respect to other variables, such as speed or direction, in addition to the positional/spatial variables? For example, how would it be possible to identify if a ship were traveling in an unusual direction and/or at an unusual speed, yet was (geographically) very close to a cluster of the training data?  Capabilities have been created to separately address how to detect anomalies in speed (see \cite{Soares} and \cite{Wen}) or anomalies in direction (see \cite{Wen}),  but many of these methods are only suited to identify specific (predefined) departures from normal behavior (see \cite{Petry}).  Spatial clustering algorithms that separately include speed or direction have even been used to identify resting or stationary points (often called ``stops") of ships.  The clustering algorithm given in Palma (\cite{Palma}), for example, identifies ``stops" as places where ships are close together and have low speed, and the clustering algorithm given in Rocha (\cite{Rocha}) identifies ``stops" as places where ships are close together and have multiple directions (if ships are at rest they can have multiple directions and not endanger one another).  Liu (\cite{Liu1} \& \cite{Liu2}) and Kontopoulos et al. (\cite{Kontopoulos}) jointly cluster on location, speed and direction in an attempt to identify anomalous behavior with respect to any and/or all of these variables.  Kontopoulos et al. (\cite{Kontopoulos}) identify anomalous behavior visually; they form convex hulls with the clusters they create, and any ship that has a trajectory which goes beyond the borders of these hulls is said to exhibit anomalous behavior.   Liu et al. (\cite{Liu2}) attempt to identify anomalous behavior by reporting a single number. The statistical significance of this number, however, is not evident. It is thus possible for Liu's statistic to report ``anomalous" behavior yet (given the training data)  not be uncommon.  This is the problem we correct for. In this paper, we offer an alternative metric which clearly conveys the statistical significance of the anomaly score, and we prove that the asymptotic distribution of this metric is normal. An additional feature to the asymptotic distribution of our proposed metric is that it is invariant to the distributions of location, speed, or direction in the training data.
 
 To explain how we create this metric, we first review how Liu et al. (\cite{Liu1}) modify the DBSCAN algorithm to create appropriate clusters. This is done in Section \ref{sctn:dbscan}. Section \ref{sctn:dbscan}  also states and clarifies some mathematical assumptions that were made in their  methodology.   Section \ref{sctn:anomalyDetection} then explains how these clusters can be used to test for anomalous behavior.  Section \ref{sctn:anomalyDetection} specifically discusses (and occasionally corrects for) how Liu (\cite{Liu2}) use these clusters to create gravity vectors and stationary sampling points. These vectors and points are summaries of the training data that are necessary in calculating an anomaly score of new/incoming trajectories.   The statistical properties of Liu's anomaly score, and those of the new anomaly score that we propose are also discussed in Section \ref{sctn:anomalyDetection}.

\section{The Data}{\label{sctn:Data}}

In this paper, we look at AIS data that was reported on January 1 of 2017 along the mid-Atlantic coast of the United States.  This data is publicly available at {\tt MarineCadastre.gov}. These data are shown in Figure \ref{fig:allData}, and this is the data set we will use for training purposes.  We will denote the training data as ${\mathcal P}$ and assume there are $n$ observations in this data set, making   $${\mathcal P} = \left \{ {\bf z}_i^{\mathcal P}: i = 1, 2, \ldots, n \right \},$$  where $${\bf z}_i^{\mathcal P} = \left( y^{\mathcal P}_i, x^{\mathcal P}_i, s^{\mathcal P}_{i}, c_{i}^{\mathcal P} \right),$$ $y^{\mathcal P}_i$ is the latitudinal position of the $i^{\rm th}$ observation, $x^{\mathcal P}_i$ is the longitudinal position of the $i^{\rm th}$ observation, $s_i^{\mathcal P}$ is the ``s"peed of the $i^{\rm th}$ observation, and $c_i^{\mathcal P}$ is the ``c"ourse of the $i^{\rm th}$ observation.\footnote{We denote this training data set as ${\mathcal P}$ since it is from this data set that we will discover ``P"atterns. We choose ${\mathcal P}$ rather than ${\mathcal T}$  (for ``T"raining) since the superscript of the elements within the set is meant to indicate what set the element is a part of.    Having a superscript of ${\mathcal T}$ may confuse the reader in thinking that a transpose is being taken.} Figure \ref{fig:allData}, of course, only shows the latitude and longitude of the training data.

{
\begin{figure}[h] 
\begin{center}
{\includegraphics[trim = 7cm 7cm 7cm 2cm, scale = 0.25]{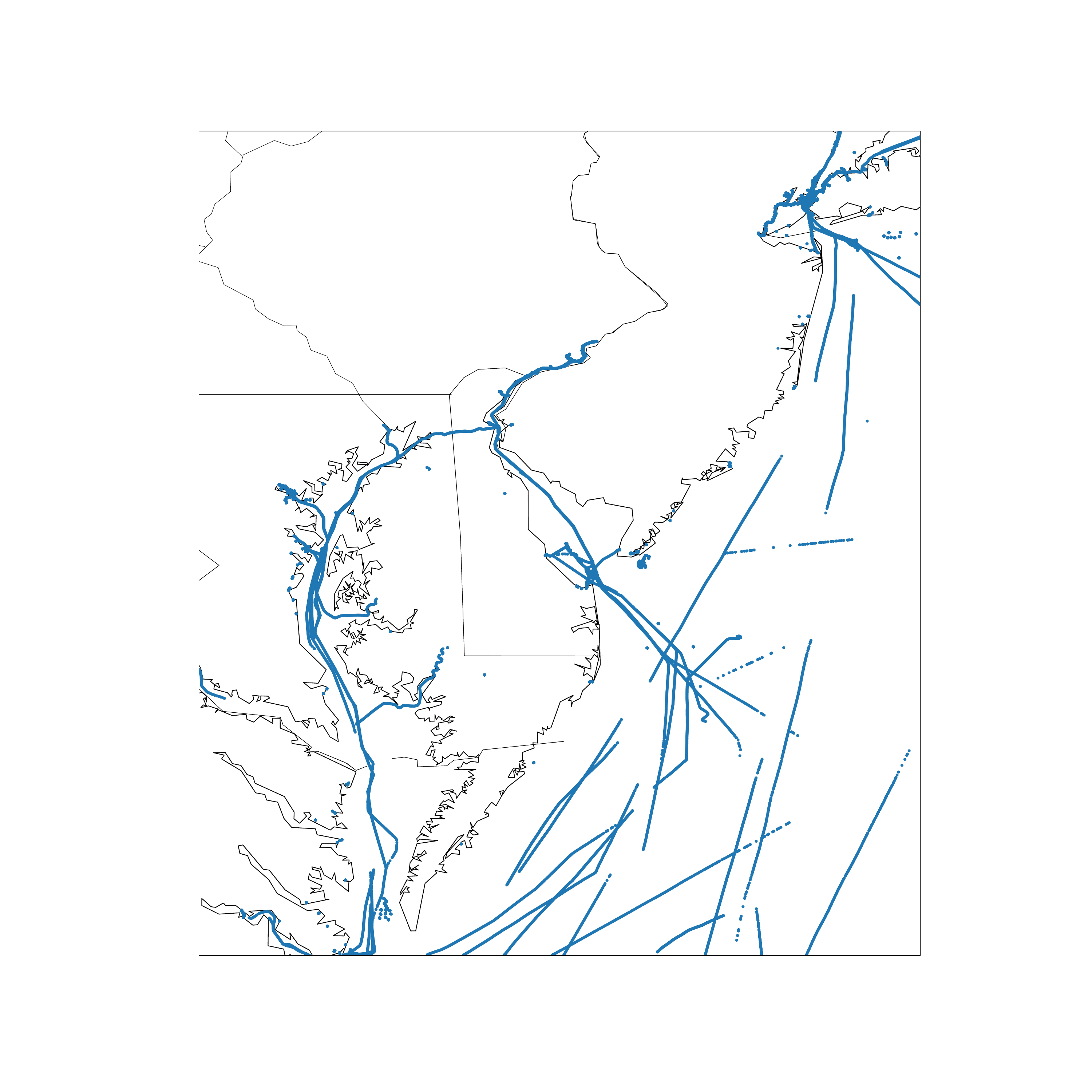}}
\caption{The Training Data, ${\mathcal P}$.}
\label{fig:allData}
\end{center}
\end{figure}
}

The goal of this paper is to identify patterns among this training data using as few assumptions as possible. Any other data which exhibits significant departures from these patterns will be considered an anomaly.

   A way to identify spatial patterns in the data while making no distributional assumptions is to apply the DBSCAN algorithm.   Although this is not the algorithm we ultimately use to create our clusters, we do use a close variation of it.   Section \ref{sctn:dbscan} reviews the DBSCAN algorithm and then provides details on how Liu et al. (\cite{Liu2}) altered it to accommodate our needs. Recall that our needs involve clustering the data not just  with respect to position, but with respect to position, speed, and course.

\section{Clustering the Data}{\label{sctn:dbscan}}

 To cluster spatial data using the DBSCAN algorithm, two parameters need to be specified, $\epsilon$ and $N_{\rm min}$. The parameter $\epsilon$ specifies the maximum distance each observation in a cluster is from another observation in that same cluster, and $N_{\rm min}$ specifies the minimum number of objects in a cluster.    The following definitions (borrowed directly from \cite{Ester} yet applied to ${\mathcal P}$) will be used to define the DBSCAN algorithm with more mathematical clarity.

\begin{description}

\item[Definition 1] A point  $ {\bf z}_i^{\mathcal P}$ is {\it directly density-reachable} from an object $ {\bf z}_j^{\mathcal P}$ with respect to $\epsilon$ and $N_{\rm min}$ in the set of objects ${\mathcal P}$ if 

\begin{enumerate}

\item $ {\bf z}_i^{\mathcal P} \in N_{\epsilon}  \left( {\bf z}_j^{\mathcal P} \right)$, where $N_{\epsilon} \left(  {\bf z}_j^{\mathcal P} \right)$ is the subset of ${\mathcal P}$ contained in the $\epsilon$ neighborhood of $ {\bf z}_j^{\mathcal P}$.

\item $| N_{\epsilon} \left(  {\bf z}_j^{\mathcal P} \right) | \geq N_{\rm min},$ where $|N_{\epsilon} \left(  {\bf z}_j^{\mathcal P} \right)|$ is the cardinality of the set $N_{\epsilon}\left(  {\bf z}_j^{\mathcal P} \right)$.

\end{enumerate}

To visualize the concept of two points being directly density-reachable, imagine the point $ {\bf z}_j^{\mathcal P}$ and all of the objects in the set ${\mathcal P}$ that are within $\epsilon$ from $ {\bf z}_j^{\mathcal P}$.  If there are at least $N_{\rm min}$ objects that are within $\epsilon$ of $ {\bf z}_j^{\mathcal P}$, and ${\bf z}_i^{\mathcal P}$ is one of those elements, then ${\bf z}_i^{\mathcal P}$ is directly density-reachable from the object $ {\bf z}_j^{\mathcal P}.$  Figure 2  illustrates objects that are directly density-reachable.

\begin{figure}[H]
{\centering

\begin{tikzpicture}

\draw (22,2) circle (2cm);
\draw[color = white, fill=blue] (22,2) circle [radius = .065];
\node [below right] at (22,2) { \textcolor{blue}{ $ {\bf z}_j^{\mathcal P}$} }; 
\draw[fill] (18,2) circle [radius = .05];
\draw[fill] (19.2,2.2) circle [radius = .05];
\draw[fill] (20.9,2.5) circle [radius = .05];
\draw[fill] (23.9,2.5) circle [radius = .05];
\draw[fill] (24.9,2.9) circle [radius = .05];
\draw[fill] (23.9,2.1) circle [radius = .05];
\draw[fill] (20.8,1.5) circle [radius = .05];
\draw[fill] (18.1,.5) circle [radius = .05];
\draw[fill] (21.9,1) circle [radius = .05];
\draw[fill] (23.1,3.71) circle [radius = .05];
\draw[fill] (24.8,0.8) circle [radius = .05];
\draw[fill] (19.1,1.0577) circle [radius = .05];
\draw[color = white, fill = red] (22.15,3.5) circle [radius = .065];
\draw[fill] (22.15,.35) circle [radius = .05];
\draw[fill] (21.95,.5) circle [radius = .05];
\draw [->] (22,2) -- (22.65, 3.85);
\node [below right] at(22.3, 2.9) { $\epsilon$};
\node[below left] at (22.15, 3.5) { \textcolor{red}{${\bf z}_i^{\mathcal P}$ }};

\end{tikzpicture}
\caption{ {\color{red} ${\bf z}_i^{\mathcal P}$} is  Directly Density-Reachable from {\color{blue} ${\bf z}_j^{\mathcal P}$} with respect to $\epsilon$ and $N_{\rm min} = 5$.}}
\label{fig:directlyDensityReachable}
\end{figure}
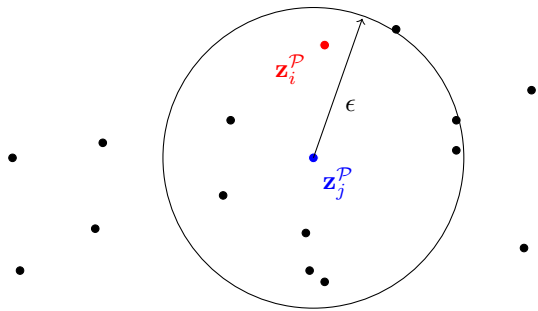

\item[Definition 2] An object ${\bf z}_i^{\mathcal P}$ is {\it density-reachable} from an object ${\bf z}_j^{\mathcal P}$ with respect to $\epsilon$ and $N_{\rm min}$ in the set of objects ${\mathcal P}$ if there is a chain of points $p_1, p_2, \ldots, p_n,$ with $p_1 = {\bf z}_j^{\mathcal P}$ and $p_n = {\bf z}_i^{\mathcal P}$ such that $p_i \in {\mathcal P}~\forall~i$ and $p_{i+1}$ is directly density-reachable from $p_i$ with respect to $\epsilon$ and $N_{\rm min}$.

Density-reachability is different from two points being directly density-reachable in the sense that it implies that there is a sequence of directly density-reachable points from ${\bf z}_j^{\mathcal P}$, and ${\bf z}_i^{\mathcal P}$ is directly density reachable from one of the points in that sequence. Figure 3  illustrates objects that are density-reachable from one-another.

\begin{figure}[H]
{\centering

\begin{tikzpicture}

\draw (22,2) circle (2cm);
\draw[color = red, fill = red] (22,2) circle [radius = .05];
\node [below right] at (22,2) { \textcolor{red}{ $ {\bf z}_j^{\mathcal P}$ }}; 
\draw [color = red][->] (22,2) -- (20.95, 3.7);
\node[below] at (21.8, 3) {\textcolor{red}{ $\epsilon$}};
\draw (23,1.9) circle (2cm);
\draw[color = red, fill = red] (23,1.9) circle [radius = .05];
\node [below] at (23, 1.9) { \textcolor{red}{ $p_2$ }};
\draw [color = red][->] (23,1.9) -- (24.8, 1.1);
\draw (24.2,2.59) circle (2cm);
\draw[color = red, fill = red] (24,2.59) circle [radius = .05];
\node [right] at (24.2, 2.59) { \textcolor{red}{ $p_3$ }};
\draw [color = red][->] (24,2.59) -- (24.4, 4.55);
\draw (30,2) circle (2cm);
\draw[color = red, fill = red] (30,2) circle [radius = .05];
\node [below right] at (30, 2) { \textcolor{red}{ $p_{7}$ }};
\draw [color = red][->] (30,2) -- (28.9, 3.65);
\draw (31.22,2.2) circle (2cm);
\draw[color = red, fill = red] (31.22,2.22) circle [radius = .05];
\node [below right] at (31.22, 2.22) { \textcolor{red}{ $p_{8}$ }};
\draw [color = red][->] (31.22,2.22) -- (33.2, 2);
\draw[fill] (23.2,3.6) circle [radius = .05];
\draw[fill] (24.4,.02) circle [radius = .05];
\draw[color = red, fill = red] (25.1,3.16) circle [radius = .05];
\node [below right] at (25.1, 3.16) { \textcolor{red}{ $p_{4}$ }};
\draw [color = red][->] (25.1,3.16) -- (27.1, 2.8);
\draw (25.1, 3.16) circle (2cm);
\draw[fill] (23.2,.01) circle [radius = .05];
\draw[fill] (23.3,2.78) circle [radius = .05];
\draw[fill] (31.6,3.7) circle [radius = .05];
\draw[fill] (28.01,1.69) circle [radius = .05];
\draw[fill] (23.8,2.54) circle [radius = .05];
\draw[fill = red, color = red] (28.9,1.53) circle [radius = .05];
\node [below] at (28.9, 1.53) { \textcolor{red}{ $p_{6}$ }};
\draw (28.9, 1.53) circle (2cm);
\draw [color = red][->] (28.9,1.53) -- (27, 1);
\draw[fill] (27.3,0.54) circle [radius = .05];
\draw[fill] (30.5,3.75) circle [radius = .05];
\draw[fill] (22.96,3.92) circle [radius = .05];
\draw[fill] (27,3.54) circle [radius = .05];
\draw[fill=red, color = red] (26.9,3.01) circle [radius = .05];
\node [above left] at (26.9, 3.01) { \textcolor{red}{ $p_{5}$ }};
\draw[fill] (26.6,2.7) circle [radius = .05];
\draw (26.9, 3.01) circle (2cm);
\draw [color = red][->] (26.9,3.01) -- (28.6, 4);

\draw[fill] (21.06,3.12) circle [radius = .05];
\draw[fill] (29.46,2) circle [radius = .05];
\draw[fill] (24.26,1.01) circle [radius = .05];
\draw[fill] (21.26,3.84) circle [radius = .05];
\draw[fill] (30.64,2.22) circle [radius = .05];
\draw[fill] (26.29,1.28) circle [radius = .05];
\draw[fill] (30.92,3.44) circle [radius = .05];
\draw[fill] (21.16,1.27) circle [radius = .05];
\draw[fill] (29.49,0.39) circle [radius = .05];
\draw[fill] (26.58,3.53) circle [radius = .05];
\draw[fill] (25.63,0.84) circle [radius = .05];
\draw[fill] (31.89,2.46) circle [radius = .05];
\draw[fill] (30.912,2.19) circle [radius = .05];
\draw[fill] (23.94,2.45) circle [radius = .05];
\draw[fill] (23.74,2.05) circle [radius = .05];
\draw[fill] (31.01,.35) circle [radius = .05];
\draw[fill] (23.69,1.77) circle [radius = .05];
\draw[fill] (29.89,1.77) circle [radius = .05];
\draw[fill] (27.44,0.34) circle [radius = .05];
\draw[fill] (22.69,0.74) circle [radius = .05];

\draw[fill] (30.26,2.55) circle [radius = .05];
\draw[fill] (31.41,3.22) circle [radius = .05];
\draw[fill] (29.68,1.03) circle [radius = .05];
\draw[fill] (29.30,0.78) circle [radius = .05];
\draw[fill] (22.12,0.35) circle [radius = .05];


\draw[color = blue, fill = blue] (32.15,3.5) circle [radius = .05];
\node[below left] at (33.15, 3.5) { \textcolor{blue}{ ${\bf z}_i^{\mathcal P}$ }};

\end{tikzpicture}
\caption{{\color{blue} $ {\bf z}_i^{\mathcal P}  $} is Density-Reachable from {\color{red} ${\bf z}_j^{\mathcal P}$} with respect to $\epsilon$ and $N_{\rm min} = 5$. In this case, $n = 9$, and observe that {\color{red} $p_{i+1}$} is directly density-reachable from {\color{red} $p_i$} for all $i$.}}
\label{fig:densityReachable}
\end{figure}
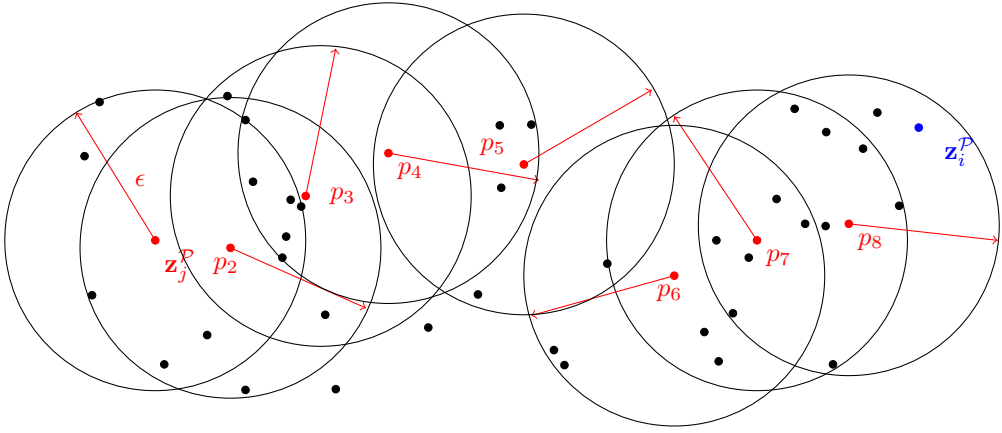

\item[Definition 3] An object ${\bf z}_i^{\mathcal P}$  is {\it density-connected} to an object ${\bf z}_j^{\mathcal P}$ with respect to $\epsilon$ and $N_{\rm min}$ in the set of objects ${\mathcal P}$ if there is  a point ${\bf z}_o^{\mathcal P} \in {\mathcal P}$ such that both ${\bf z}_i^{\mathcal P}$ and ${\bf z}_j^{\mathcal P}$ are density-reachable from ${\mathcal P}$ with respect to $\epsilon$ and $N_{\rm min}$ in ${\mathcal P}$.

Figure 4  illustrates the concept of two elements in ${\mathcal P}$ being density-connected.

\begin{figure}[H]
{\centering

\begin{tikzpicture}

\draw (22,2) circle (2cm);
\draw[color = red, fill = red] (22,2) circle [radius = .05];
\draw [color = red][->] (22,2) -- (20.95, 3.7);
\node[below] at (21.8, 3) {\textcolor{red}{ $\epsilon$}};
\draw (23,1.9) circle (2cm);
\draw[color = red, fill = red] (23,1.9) circle [radius = .05];
\draw [color = red][->] (23,1.9) -- (24.8, 1.1);
\draw (24.2,2.59) circle (2cm);
\draw[color = red, fill = red] (24,2.59) circle [radius = .05];
\node [below right] at (24, 2.59) { \textcolor{red}{ ${\bf z}_o^{\mathcal P}$ }};
\draw [color = red][->] (24,2.59) -- (24.4, 4.55);
\draw (30,2) circle (2cm);
\draw[color = red, fill = red] (30,2) circle [radius = .05];
\draw [color = red][->] (30,2) -- (28.9, 3.65);
\draw (31.22,2.2) circle (2cm);
\draw[color = red, fill = red] (31.22,2.22) circle [radius = .05];
\draw [color = red][->] (31.22,2.22) -- (33.2, 2);
\draw[fill] (23.2,3.6) circle [radius = .05];
\draw[fill] (24.4,.02) circle [radius = .05];
\draw[color = red, fill = red] (25.1,3.16) circle [radius = .05];
\draw [color = red][->] (25.1,3.16) -- (27.1, 2.8);
\draw (25.1, 3.16) circle (2cm);
\draw[fill] (23.2,.01) circle [radius = .05];
\draw[fill] (23.3,2.78) circle [radius = .05];
\draw[fill] (31.6,3.7) circle [radius = .05];
\draw[fill] (28.01,1.69) circle [radius = .05];
\draw[fill] (23.8,2.54) circle [radius = .05];
\draw[fill = red, color = red] (28.9,1.53) circle [radius = .05];
\draw (28.9, 1.53) circle (2cm);
\draw [color = red][->] (28.9,1.53) -- (27, 1);
\draw[fill] (27.3,0.54) circle [radius = .05];
\draw[fill] (30.5,3.75) circle [radius = .05];
\draw[fill] (22.96,3.92) circle [radius = .05];
\draw[fill] (27,3.54) circle [radius = .05];
\draw[fill=red, color = red] (26.9,3.01) circle [radius = .05];
\draw[fill] (26.6,2.7) circle [radius = .05];
\draw (26.9, 3.01) circle (2cm);
\draw [color = red][->] (26.9,3.01) -- (28.6, 4);

\draw[fill] (21.06,3.12) circle [radius = .05];
\draw[fill] (29.46,2) circle [radius = .05];
\draw[fill] (24.26,1.01) circle [radius = .05];
\draw[fill] (21.26,3.84) circle [radius = .05];
\draw[fill] (30.64,2.22) circle [radius = .05];
\draw[fill] (26.29,1.28) circle [radius = .05];
\draw[fill] (30.92,3.44) circle [radius = .05];
\draw[fill] (21.16,1.27) circle [radius = .05];
\draw[fill] (29.49,0.39) circle [radius = .05];
\draw[fill] (26.58,3.53) circle [radius = .05];
\draw[fill] (25.63,0.84) circle [radius = .05];
\draw[fill] (31.89,2.46) circle [radius = .05];
\draw[fill] (30.912,2.19) circle [radius = .05];
\draw[fill] (23.94,2.45) circle [radius = .05];
\draw[fill] (23.74,2.05) circle [radius = .05];
\draw[fill] (31.01,.35) circle [radius = .05];
\draw[fill] (23.69,1.77) circle [radius = .05];
\draw[fill] (29.89,1.77) circle [radius = .05];
\draw[fill] (27.44,0.34) circle [radius = .05];
\draw[fill] (22.69,0.74) circle [radius = .05];

\draw[fill] (30.26,2.55) circle [radius = .05];
\draw[fill] (31.41,3.22) circle [radius = .05];
\draw[fill] (29.68,1.03) circle [radius = .05];
\draw[fill] (29.30,0.78) circle [radius = .05];
\draw[color = blue, fill = blue] (22.12,0.35) circle [radius = .05];

\node [above right] at (22.12,.35) { \textcolor{blue}{ $ {\bf z}_j^{\mathcal P}$ }};

\draw[color = blue, fill = blue] (32.15,3.5) circle [radius = .05];
\node[below left] at (33.15, 3.5) { \textcolor{blue}{ ${\bf z}_i^{\mathcal P}$ }};

\end{tikzpicture}
\caption{{\color{blue} $ {\bf z}_j^{\mathcal P}  $} and {\color{blue} ${\bf z}_i^{\mathcal P}$} are  Density-Connected with respect to $\epsilon$ and $N_{\rm min} = 5$. Observe that both are Density-Reachable from {\color{red} ${\bf z}_0^{\mathcal P}$}}}
\label{fig:Hello}
\end{figure}
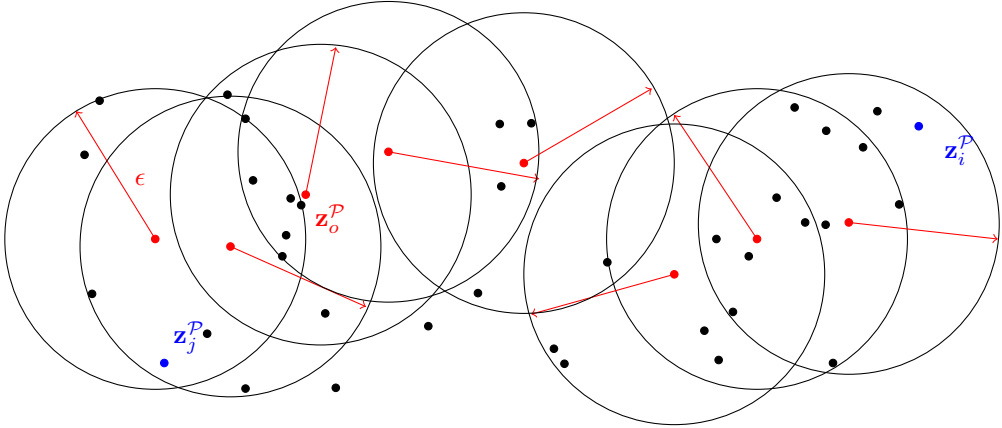

\end{description}

In the DBSCAN algorithm, all elements in the same cluster are density-connected. Another (and perhaps more intuitive) way to communicate this is to say that for every element in a DBSCAN cluster, there is at least one other element in that same cluster which is $\epsilon$ units away. And it must be the case that for at least one of the elements that is $\epsilon$ units away, there are at least $N_{\rm min}$ observations within $\epsilon$ of it.

Applying the DBSCAN algorithm to the training data shown in Figure \ref{fig:allData} with $\epsilon = .02$,  $N_{\rm min} = 5$, and ${\rm dist} \left( {\bf z}_i^{\mathcal P}, {\bf z}_j^{\mathcal P} \right)$ being $${\rm dist} \left( {\bf z}_i^{\mathcal P}, {\bf z}_j^{\mathcal P} \right) =    \left \|   \left( y_{i}^{\mathcal P},~x_{i}^{\mathcal P}  \right)^T - \left(   y_{j}^{\mathcal P},~x_{j}^{\mathcal P}  \right)^T     \right \|,$$ we get the results shown in Figure \ref{fig:DBSCANResults}. Note that the distance we calculate between two points is the Euclidian distance between the two points' latitude and longitude coordinates.   Our methodology thus assumes a flat earth. The more locally we apply our algorithm, the more valid this assumption.

{
\begin{figure}[H] 
\begin{center}
{\includegraphics[trim = 7cm 7cm 7cm 2cm, scale = 0.25]{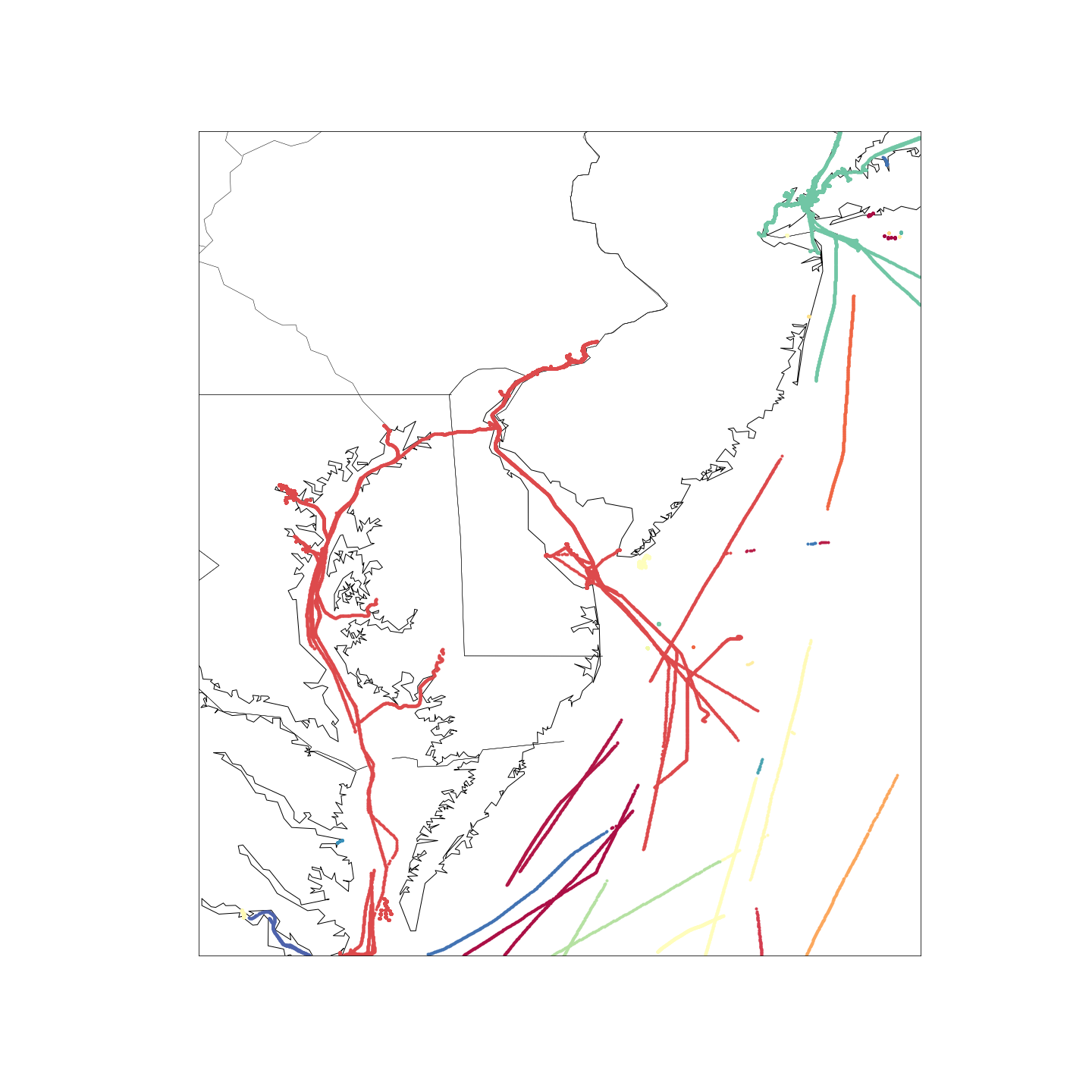}}
\caption{DBSCAN Results. The different clusters are in different colors.}
\label{fig:DBSCANResults}
\end{center}
\end{figure}
}

The results in Figure \ref{fig:DBSCANResults} cluster points based on their location. We wish to cluster with respect to location, speed, and direction, however.  Let us specifically assume that we wish to cluster such that points ${\bf z}_i^{\mathcal P}$ and ${\bf z}_j^{\mathcal P}$ are in the same cluster if ${\rm dist} \left( {\bf z}_i^{\mathcal P}, {\bf z}_j^{\mathcal P} \right) < \epsilon_{\rm Dist}$, $ \left| c_i^{\mathcal P} - c_j^{\mathcal P} \right| < \epsilon_{\rm Crs}$, and $ \left| s_i^{\mathcal P} - s_j^{\mathcal P} \right| < \epsilon_{\rm Spd}. $   To do this, an  edited version of the DBSCAN algorithm is necessary since DBSCAN is only engineered to consider one distance between two points.    Liu et al. (\cite{Liu1}) edit the DBSCAN algorithm to accommodate these wishes. Their edited version of the DBSCAN algorithm is called DBSCANSD, where the ``SD" stands for ``S"peed and ``D"irection.

The DBSCANSD  algorithm applied to the training data set ${\mathcal P}$ is given below and requires that the additional thresholds $\epsilon_{\rm Crs}$ and $\epsilon_{\rm Spd}$ be specified.   Liu et al.'s (\cite{Liu1}) addition to the standard DBSCAN algorithm is shown in red.

\begin{algorithm}

	\SetKwInOut{Input}{input}
	\SetKwInOut{Output}{output}
	\SetKwFunction{QueryNeighborPoints}{QueryNeighborPoints}
	\SetKwFunction{DBSCANSD}{DBSCANSD}
	\SetKwFunction{add}{add}
	\SetKwFunction{MergeClusters}{MergeClusters}
	\SetKwFunction{remove}{remove}
	\SetKwFunction{distance}{distance}
	
\hspace{-.4cm} {{\bf Procedure}: \DBSCANSD}
	
	\BlankLine

\Input{${\mathcal P}$,  $N_{\rm min}$, $\epsilon_{\rm Dist}$, $\epsilon_{\rm Crs}$, $\epsilon_{\rm Spd}$ }
	\Output{$cltrList$}

\BlankLine

	$cltrList \leftarrow$ empty list

	\For{\rm each unclassified point ${\bf z}_i^{\mathcal P} \in {\mathcal P} $}{
		Mark $ {\bf z}i^{\mathcal P}$ as classified

		$neighborPts \leftarrow$ \QueryNeighborPoints($ {\mathcal P}$, $ {\bf z}_i^{\mathcal P}$, $N_{\rm min}$, $\epsilon_{\rm Dist}$, $\epsilon_{\rm Crs}$, $\epsilon_{\rm Spd}$)

		\If{\rm $neighborPts$ is not NULL} {
			$cltrList$.\add(neighborPts)	}

		\For {\rm each cluster $C$ in $cltrList$} { 

			\For {\rm each cluster $C'$ in $cltrList$} {

				\If {\rm $C$ and $C'$ are different clusters} {

					\If {\MergeClusters ($C$, $C'$) is {\tt TRUE}} {
				
						$cltrList$.\remove($C'$) 

					}
}
}
}
}

\BlankLine

\BlankLine

\hspace{-.4cm} {{\bf Procedure}: \QueryNeighborPoints}

\BlankLine

\Input{$ {\mathcal P}$, $ {\bf z}_i^{\mathcal P}$, $N_{\rm min}$, $\epsilon_{\rm Dist}$, $\epsilon_{\rm Crs}$, $\epsilon_{\rm Spd}$}
\Output{$cluster$}

\BlankLine

$cluster \leftarrow$ empty list

\For {\rm each point $ {\bf X}_j^{\mathcal P}$ in  ${\mathcal X}^{\mathcal P}$ } {

\If {\distance $\left( \left( y_{i}^{\mathcal P}, x_{i}^{\mathcal P} \right)^T,  \left( y_{j}^{\mathcal P}, x_{j}^{\mathcal P} \right)^T \right)  < \epsilon_{\rm Dist}$} {

 \If { {\color{red} $\left| c_i^{\mathcal P} - c_j^{\mathcal P} \right| < \epsilon_{\rm Crs}$}} {

\If { {\color{red} $ \left| s_i^{\mathcal P} - s_j^{\mathcal P} \right| < \epsilon_{\rm Spd}$}} {

  $cluster$.\add( ${\bf z}_j^{\mathcal P} $ ) }}

}

}

\If {\rm $cluster$.size  $ \geq N_{\rm min}$} {

Mark $ {\bf z}_i^{\mathcal P}$ as core point

}

\BlankLine

\BlankLine

\hspace{-.4cm}{\bf Procedure}: \MergeClusters

\BlankLine

\Input{$clusterA$, $clusterB$}
\Output{$merge$}

$merge \leftarrow {\tt FALSE}$

\For {\rm each point $ {\bf z}_j^{\mathcal P}$ in $clusterB$}  {

\If{ {\rm point $ {\bf z}_j^{\mathcal P}$ is core point and $clusterA$ contains cluster $Q$}} {

	$merge \leftarrow $ {\tt TRUE}

\For {\rm each point $ {\bf z}_l^{\mathcal P}$ in $clusterB$}  {

	$clusterA$.\add($ {\bf z}_l^{\mathcal P}$)
}
}

}

\caption{DBSCANSD}
\end{algorithm}

We apply the DBSCANSD algorithm to the ``moving" members of ${\mathcal P}$, ${\mathcal P}^{\rm mv}$, and the DBSCAN algorithm to the stationary members of ${\mathcal P}$, ${\mathcal P}^{\rm st}$ to get the results shown in Figure \ref{fig:DBSCANSDResults}. It should be noted that in this application, the speed threshold used was 2.5 knts and the direction threshold used was 90 degrees ($\epsilon_{\rm Spd} = 2.5~ {\rm knts}$ and $\epsilon_{\rm Crs} = 90~{\rm degrees}$)

{
\begin{figure}[H] 
\begin{center}
{\includegraphics[trim = 7cm 7cm 7cm 2cm, scale = 0.25]{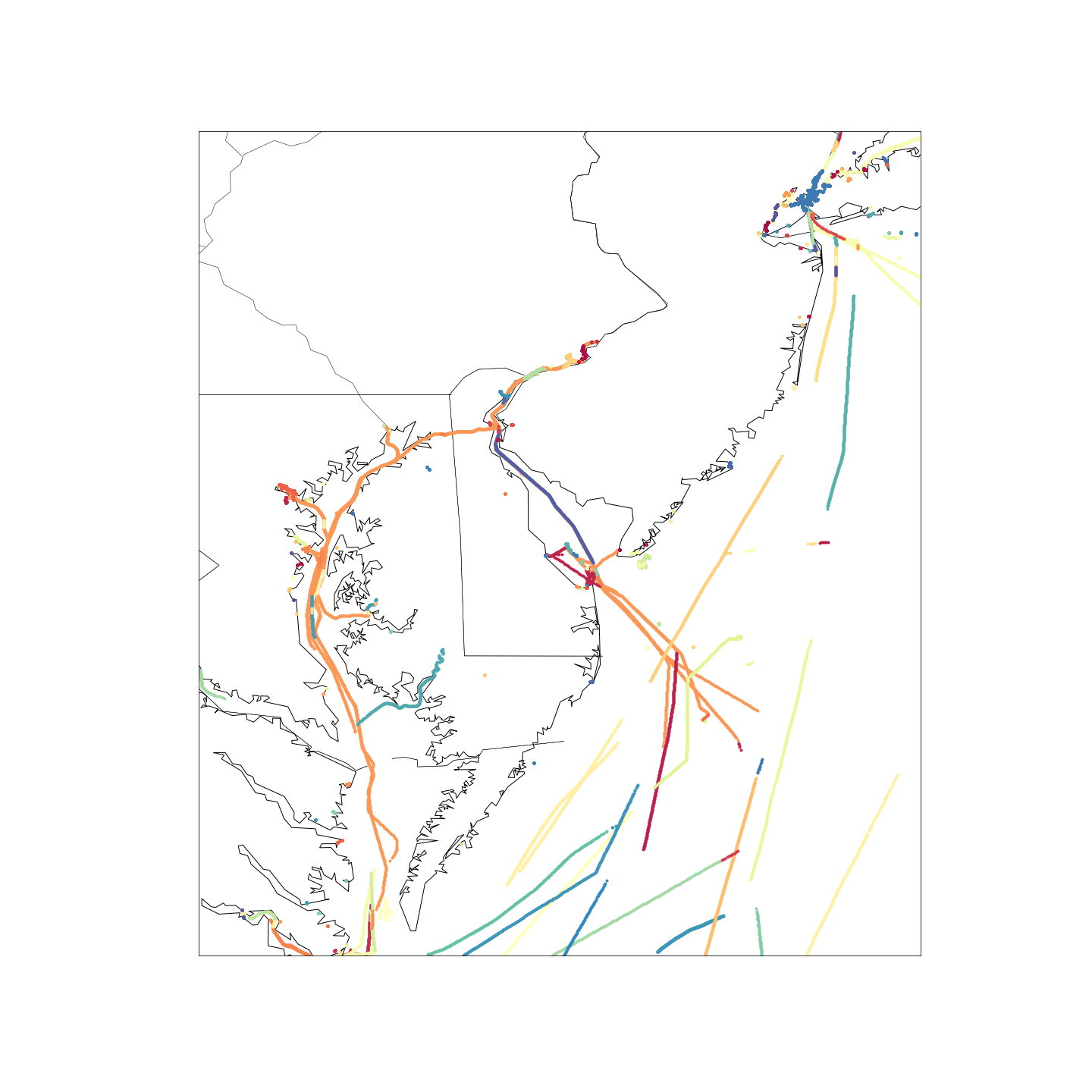}}
\caption{DBSCANSD Results. The different clusters are in different colors.}
\label{fig:DBSCANSDResults}
\end{center}
\end{figure}
}

From this point on, we will denote the set of points put in cluster $j$ as a result of the DBSCANSD algorithm as  ${\mathcal P}^{{\rm mv},~{\rm cl}~j},$ and we will assume that the number of observations in this cluster is $n_j^{\rm mv}$, making $${\mathcal P}^{{\rm mv},~{\rm cl}~j} = \left \{ {\bf z}_1^{{\rm mv},~{\rm cl}~j},  {\bf z}_2^{{\rm mv},~{\rm cl}~j}, \ldots,  {\bf z}_{n_j^{\rm mv}-1}^{{\rm mv},~{\rm cl}~j},  {\bf z}_{n_j^{\rm mv}}^{{\rm mv},~{\rm cl}~j} \right \}.$$ The set of points put in cluster $l$ as a result of the DBSCAN algorithm will be denoted as ${\mathcal P}^{{\rm st,~cl}~l},$ and we will assume that the number of observations in this cluster is $n_l^{\rm st}$, making $${\mathcal P}^{{\rm st},~{\rm cl}~l} = \left \{ {\bf z}_1^{{\rm st},~{\rm cl}~l},  {\bf z}_2^{{\rm st},~{\rm cl}~l}, \ldots,  {\bf z}_{n_l^{\rm st}-1}^{{\rm st},~{\rm cl}~l},  {\bf z}_{n_l^{\rm st}}^{{\rm st},~{\rm cl}~l} \right \}.$$

The next section discusses how we use these results to identify outliers or outlying/anomalous behavior in ships.

\section{Identifying Anomalous Behavior}{\label{sctn:anomalyDetection}}

To identify anomalous behavior in a new ship's trajectory, one first has to separate this new trajectory into a set of stationary points and a set of moving points. Just as Liu et. al (\cite{Liu2}) did, we identify the set of stationary points as that set such that the speed is less than 0.5 knots. The set of moving points is the complement of that. We will assume that there are a total of $m$ points in this new trajectory, and we will denote the set of points in this trajectory as ${\mathcal N}$ (for ``N"ew), where  $${\mathcal N} = \left \{ {\bf z}_1^{\mathcal N}, {\bf z}_2^{\mathcal N}, \ldots, {\bf z}_m^{\mathcal N} \right \}.$$  The set of points in ${\mathcal N}$ that are moving is $${\mathcal N}^{\rm mv} = \left \{ {\bf  z}_i^{\mathcal N} : s_{i}^{\mathcal N} \geq 0.5~{\rm knts} \right \},$$ and the set of stationary points is  $${\mathcal N}^{\rm st} = \left \{ {\bf z}_i^{\mathcal N} : s_{i}^{\mathcal N} < 0.5~{\rm knts} \right \}.$$ We will assume there are $m^{\rm mv}$ values in ${\mathcal N}^{\rm mv}$ and $m^{\rm st}$ values in ${\mathcal N}^{\rm st}$, making $$ {\mathcal N}^{\rm mv} = \left \{ {\bf z}_1^{{\mathcal N}^{\rm mv}}, {\bf z}_2^{{\mathcal N}^{\rm mv}}, \ldots, {\bf z}_{m^{\rm mv} - 1}^{{\mathcal N}^{\rm mv}}, {\bf z}_{m^{\rm mv}}^{{\mathcal N}^{\rm mv}} \right \},$$ and $$ {\mathcal N}^{\rm st} = \left \{ {\bf z}_1^{{\mathcal N}^{\rm st}}, {\bf z}_2^{{\mathcal N}^{\rm st}}, \ldots, {\bf z}_{m^{\rm st} - 1}^{{\mathcal N}^{\rm st}}, {\bf z}_{m^{\rm st}}^{{\mathcal N}^{\rm st}} \right \}.$$

We then see how the points in ${\mathcal N}^{\rm st}$ depart from stationary points in the training data set, and how the points in ${\mathcal N}^{\rm mv}$ depart from the moving points in the training data set. To do this, Liu et al. (\cite{Liu2}) first create two sets of points, one set which summarizes the stationary points in the training data, the other set which summarizes the moving points in the training data. These sets are respectively called the stationary sampling points and gravity vectors, and we denote these sets as ${\mathcal S}$ and ${\mathcal G}$. The set of new trajectory points, ${\mathcal N}$, are then compared to ${\mathcal S}$ and ${\mathcal G}$ (the set ${\mathcal N}^{\rm st}$ is compared to ${\mathcal S}$, and the set ${\mathcal N}^{\rm mv}$ is compared to ${\mathcal G}$), and it is from this comparison that a trajectory is identified as being anomalous or not.   

Subsection \ref{sctn:CalcGravAndStat} describes how the gravity vectors and stationary sampled points are calculated. Subsection \ref{sctn:CalcAnom} describes how these set of points, ${\mathcal S}$ and ${\mathcal G}$, are compared to the new trajectory, ${\mathcal N}$. The subsection specifically reviews how Liu et al. (\cite{Liu2}) calculate and assign an anomalous score to a new trajectory and then discusses our alternative anomalous score.

\subsection{Creating Gravity Vectors and Stationary Sampled Points}\label{sctn:CalcGravAndStat}

Generally speaking, a gravity vector is a point (or vector) that is meant to summarize and describe all of the moving points of the training set around it.  We will let ${\mathcal G}^{{\rm cl}~j}$ be the set of gravity vectors which summarize cluster $j$, and we will assume there are $n^{{\rm cl}~j}_{\rm grv}$ gravity vectors associated with cluster $j$, i.e.,  $${\mathcal G}^{{\rm cl}~j} = \left \{ {\bf g}_1^{{\rm cl}~j}, {\bf g}_2^{{\rm cl}~j}, \ldots, {\bf g}_{n^{{\rm cl}~j}_{\rm grv} - 1}^{{\rm cl}~j}, {\bf g}_{n^{{\rm cl}~j}_{\rm grv}}^{{\rm cl}~j} \right \}.  $$

  To calculate all of the  gravity vectors associated with cluster $j$, one first has to calculate the average course (direction) of the entire cluster.  We will call this average direction ${\overline c}^{{\rm mv},~{\rm cl}~j}$ and calculate it as $${\overline c}^{{\rm mv},~{\rm cl}~j} = {\frac{1}{n_j^{\rm mv}}} \sum_{k=1}^{n_j^{\rm mv}}c^{{\rm mv},~{\rm cl}~j}_k.$$     After calculating ${\overline c}^{{\rm mv},~{\rm cl}~j} $, one must define a line along this direction and divide this into segments of length $\delta$. Liu et al. (\cite{Liu2}) set $\delta = \epsilon$, and we do the same to achieve our results. All of the observations that are then within a particular  band (of width $\delta$) are considered and their average latitude, longitude, speed, direction, and median distance from the average position are calculated and reported in one gravity vector. This concept is illustrated in Figure 7, and in this paper we mathematically formalize the calculation of these vectors. This is done in the text box below.

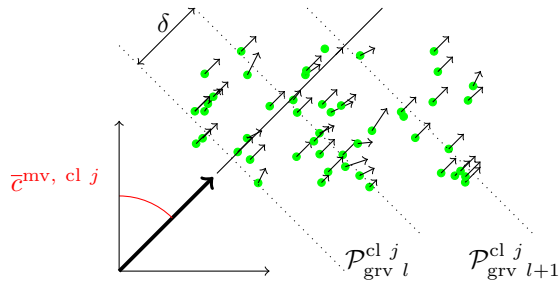
\begin{figure}[h!]\label{fig:GravVec}
{\centering

\begin{tikzpicture}

\draw[color = green, fill = green] (23.63800, 1.2710323) circle [radius = .05];
\draw[->] (23.63800, 1.2710323)  -- (23.83800, 1.4510323);
\draw[color = green, fill = green] (20.70091, 0.7684835) circle [radius = .05];
\draw[->] (20.70091, 0.7684835)  -- (20.90091, 0.8684835);
\draw[color = green, fill = green] (22.19662, 0.2706696) circle [radius = .05];
\draw[->] (22.19662, 0.2706696)  -- (22.39662, 0.3706696);
\draw[color = green, fill = green] (21.73469, 1.9576310) circle [radius = .05]; 
\draw[->] (21.87736,  1.8768501)  -- (21.97736,  1.9768501);
\draw[color = green, fill = green] (21.31523,   1.2765258) circle [radius = .05];
\draw[->] (21.31523,   1.2765258)  -- (21.41523,   1.4065258);
\draw[color = green, fill = green] (21.47661, 1.6106625) circle [radius = .05]; 
\draw[->] (21.47661, 1.6106625)  -- (21.67661, 1.7506625);
\draw[color = green, fill = green] (20.70571,  1.6106625) circle [radius = .05];
\draw[->] (20.70571,  1.6106625)  -- (20.85571,  1.9106625);
[\draw[color = green, fill = green] (22.00689, 0.3890405) circle [radius = .05]; 
[\draw[->] (22.00689, 0.3890405)  -- (22.30689, 0.4890405);
\draw[color = green, fill = green] (20.84811,  0.1728718) circle [radius = .05];
\draw[->] (20.84811,  0.1728718)  -- (20.94811,  0.3728718);
\draw[color = green, fill = green] (22.20429,  1.8689545) circle [radius = .05];
\draw[->] (22.20429,  1.8689545)  -- (22.40429,  1.968954);
\draw[color = green, fill = green] (20.74048,  0.4838273) circle [radius = .05];
\draw[->] (20.74048,  0.4838273)  -- (20.94048,  0.6838273);
\draw[color = green, fill = green] (20.57647, 0.5855337) circle [radius = .05];
\draw[->] (20.57647, 0.5855337)  -- (20.77647, 0.7855337);
\draw[color = green, fill = green] (22.36117, 0.8679415) circle [radius = .05];
\draw[->] (22.36117, 0.8679415)  -- (22.56117, 1.1679415);
\draw[color = green, fill = green] (21.48883,   1.6680491) circle [radius = .05];
\draw[->] (21.48883,   1.6680491)  -- (21.68883,   1.868049);
\draw[color = green, fill = green] (20.13777, 1.6244667) circle [radius = .05];
\draw[->] (20.13777, 1.6244667)  -- (20.33777, 1.8244667);
\draw[color = green, fill = green] (23.46295,  0.2691778) circle [radius = .05];
\draw[->] (23.46295,  0.2691778)  -- (23.66295,  0.4691778);
\draw[color = green, fill = green] (21.35790,  0.5222948) circle [radius = .05];
\draw[->] (21.35790,  0.5222948)  -- (21.55790,  0.7222948);
\draw[color = green, fill = green] (23.23690, 1.6532678) circle [radius = .05];
\draw[->] (23.23690, 1.6532678)  -- (23.43690, 1.853267);
\draw[color = green, fill = green] (21.00206, 1.1943203) circle [radius = .05];
\draw[->] (21.00206, 1.1943203)  -- (21.20206, 1.3943203);
\draw[color = green, fill = green] (20.13530, 1.1258566) circle [radius = .05];
\draw[->] (20.13530, 1.1258566)  -- (20.33530, 1.4258566);
\draw[color = green, fill = green] (21.66604, 0.5559928) circle [radius = .05];
\draw[->] (21.66604, 0.5559928)  -- (21.86604, 0.7559928);
\draw[color = green, fill = green] (21.97297, 0.5465695) circle [radius = .05];
\draw[->] (21.97297, 0.5465695)  -- (22.17297, 0.7465695);
\draw[color = green, fill = green] (22.17187, 0.6968539) circle [radius = .05];
\draw[->] (22.17187, 0.6968539)  -- (22.37187, 0.7168539);
\draw[color = green, fill = green] (20.62276, 1.9244634) circle [radius = .05];
\draw[->] (20.62276, 1.9244634)  -- (20.82276, 2.124463);
\draw[color = green, fill = green] (21.66326, 0.8349072) circle [radius = .05];
\draw[->] (21.66326, 0.8349072)  -- (21.86326, 0.9349072);
\draw[color = green, fill = green] (21.70014, 1.2180924) circle [radius = .05];
\draw[->] (21.70014, 1.2180924)  -- (21.90014, 1.4180924);
\draw[color = green, fill = green] (20.18003, 1.2264357) circle [radius = .05];
\draw[->] (20.18003, 1.2264357)  -- (20.38003, 1.4264357);
\draw[color = green, fill = green] (20.25083, 1.3200750) circle [radius = .05];
\draw[->] (20.25083, 1.3200750)  -- (20.45083, 1.5200750);
\draw[color = green, fill = green] (22.93654, 0.7740642) circle [radius = .05];
\draw[->] (22.93654, 0.7740642)  -- (23.13654, 0.9740642);
\draw[color = green, fill = green] (21.81469, 1.1133060) circle [radius = .05];
\draw[->] (21.81469, 1.1133060)  -- (22.11469, 1.2533060);
\draw[color = green, fill = green] (23.18588, 1.9233857) circle [radius = .05];
\draw[->] (23.18588, 1.9233857)  -- (23.38588, 2.1233857);
\draw[color = green, fill = green] (22.74904, 1.1178792) circle [radius = .05];
\draw[->] (22.74904, 1.1178792)  -- (22.94904, 1.3178792);
\draw[color = green, fill = green] (23.60088, 0.1814350) circle [radius = .05];
\draw[->] (23.60088, 0.1814350)  -- (23.80088, 0.3814350);
\draw[color = green, fill = green] (21.58943, 0.7273508) circle [radius = .05];
\draw[->] (21.58943, 0.7273508)  -- (21.78943, 0.9273508);
\draw[color = green, fill = green] (23.53825, 0.3317697) circle [radius = .05];
\draw[->] (23.53825, 0.3317697)  -- (23.73825, 0.5317697);
\draw[color = green, fill = green] (20.10950, 0.7668967) circle [radius = .05];
\draw[->] (20.10950, 0.7668967)  -- (20.30950, 0.9668967);
\draw[color = green, fill = green] (23.38196, .6255913) circle [radius = .05];
\draw[->] (23.38196, .6255913)  -- (23.58196, .8255913);
\draw[color = green, fill = green] (23.28079, 0.3049523) circle [radius = .05];
\draw[->] (23.28079, 0.3049523)  -- (23.48079, 0.504952);
\draw[color = green, fill = green] (21.36675, 1.1100533) circle [radius = .05];
\draw[->] (21.36675, 1.1100533)  -- (21.56675, 1.3100533);
\draw[color = green, fill = green] (21.77827, 0.3198541) circle [radius = .05];
\draw[->] (21.77827, 0.3198541)  -- (21.97827, 0.519854);
\draw[color = green, fill = green] (23.60390, 0.2577688) circle [radius = .05];
\draw[->] (23.60390, 0.2577688)  -- (23.80390, 0.4577688);
\draw[color = green, fill = green] (20.02340, 0.6954620) circle [radius = .05];
\draw[->] (20.02340, 0.6954620)  -- (20.22340, 0.895462);
\draw[color = green, fill = green] (23.16879, 1.2487912) circle [radius = .05];
\draw[->] (23.16879, 1.2487912)  -- (23.36879, 1.4487912);
\draw[color = green, fill = green] (21.95186, 1.2029393) circle [radius = .05];
\draw[->] (21.95186, 1.2029393)  -- (22.15186, 1.3529393);
\draw[color = green, fill = green] (20.00662, 1.1266327) circle [radius = .05];
\draw[->] (20.00662, 1.1266327)  -- (20.20662, 1.3266327);
\draw[color = green, fill = green] (22.78272, 1.0562798) circle [radius = .05];
\draw[->] (20.00662, 1.1266327)  -- (20.20662, 1.3266327);
\draw[color = green, fill = green] (23.70571, 1.4627505) circle [radius = .05];
\draw[->] (23.70571, 1.4627505)  -- (23.80571, 1.6627505);
\draw[color = green, fill = green] (22.32901, 0.1183555) circle [radius = .05];
\draw[->] (22.32901, 0.1183555)  -- (22.42901, 0.2183555);
\draw[color = green, fill = green] (21.69016, 0.1783267) circle [radius = .05];
\draw[->] (21.69016, 0.1783267)  -- (21.79016, 0.2783267);
\draw[->, line width = .05cm] (19, -1) -- (20.25, 0.25);
\draw [red] (19.7071,-.2929) arc [radius = 1, start angle = 45, end angle = 90];
\draw[->](19,-1) -- (19, 1);
\draw[->](19,-1) -- (21,-1);
\draw[<->](20, 2.5) -- (19.25, 1.75);
\draw [dotted] (19,2) -- (22, -1);
\draw [dotted] (20, 2.5) -- (   23  ,-.5);
\draw (20.3, 0.3) -- (22.5, 2.5);
\draw [dotted] (21.5, 2.5) -- (24.5,  -.5);
\node [left] at (19.8, 2.3) {$\delta$};
\node [below left] at (19, .5) { \textcolor{red}{ ${\overline c}^{{\rm mv,~cl}~j}$ }};
\node [below left] at (23, -.5) {$ {\mathcal P}^{{\rm cl}~j}_{{\rm grv}~l}$};
\node [below left] at (25, -.5) {$ {\mathcal P}^{{\rm cl}~j}_{{\rm grv}~l+1}$};

\end{tikzpicture}
\caption{Calculating Gravity Vectors}}
\end{figure}

\begin{tcolorbox} 

\begin{center}

{\bf Calculating Components of ${\bf g}_l^{{\rm cl}~j}$}

\end{center}

 The line along the direction of ${\overline c}^{{\rm mv},~{\rm cl}~j}$ is of length $L^{{\rm mv},~{\rm cl}~j}$, where $$L^{{\rm mv},~{\rm cl}~j} =  \left \{  \max_i \left[ y_{i}^{{\rm mv},~{\rm cl}~j} / \cos \left( {\overline c}^{{\rm mv},~{\rm cl}~j} \right) \right] - \min_i \left[ y_{i}^{{\rm mv},~{\rm cl}~j} / \cos \left( {\overline c}^{{\rm mv},~{\rm cl}~j} \right) \right] \right \}.$$ We are dividing this line into bands of width $\delta$, making the number of gravity vectors in cluster $j$,  $n^{{\rm cl}~j}_{\rm grv},$  $$n^{{\rm cl}~j}_{\rm grv}  = \left. L^{{\rm mv},~{\rm cl}~j} \right/ \delta.$$    The observations in cluster $j$  to be considered in the calculation of the $l^{\rm th}$ gravity $\left( 1 \leq l \leq n_{\rm grv}^{{\rm cl}~j} \right)$  vector are then $${\mathcal P}^{{\rm cl}~j}_{{\rm grv}~l} = \left \{   {\bf z}_i^{{\rm mv},~{\rm cl}~j}  : (l-1) \cdot b \cdot \delta \leq  \left.  y_{i}^{{\rm mv},~{\rm cl}~j} \right/ \cos \left( {\overline c}^{{\rm mv},~{\rm cl}~j} \right) \leq l \cdot b \cdot \delta \right \},$$ where $b = \min_i \left(  \left. y_{i}^{{\rm mv},~{\rm cl}~j}  \right/ \cos \left(   {\overline c}^{{\rm mv},~{\rm cl}~j} \right) \right).$ In words, the observations in ${\mathcal P}_{{\rm grv}~l}^{{\rm cl}~j}$ are those in the $l^{\rm th}$ slice of the moving points in cluster $j$. The direction and orientation of this slice is governed by ${\rm cos} \left( {\overline c}^{{\rm mv,~cl}~j} \right).$ The $l^{\rm th}$ gravity vector in cluster $j$ is then the average location, speed, and course over this set. It also includes the median distance of each observation to the location of the gravity vector. We denote the $l^{\rm th}$ gravity vector in cluster $j$ as $${\bf g}_l^{{\rm cl}~j}  =  \left( y_{{\rm grv}~l}^{{\rm cl}~j}, ~  x_{{\rm grv}~l}^{{\rm cl}~j},~  s_{{\rm grv}~l}^{{\rm cl}~j}, ~ c_{{\rm grv}~l}^{{\rm cl}~j}, ~ d_{{\rm grv}~l}^{{\rm cl}~j}    \right)^T,  $$  where $$  \left( y_{{\rm grv}~l}^{{\rm cl}~j},   x_{{\rm grv}~l}^{{\rm cl}~j},  s_{{\rm grv}~l}^{{\rm cl}~j},  c_{{\rm grv}~l}^{{\rm cl}~j}  \right)^T =   \left|  {\mathcal P}^{{\rm cl}~j}_{{\rm grv}~l}       \right|^{-1}\sum_{  {\bf z}_i^{\mathcal P} \in    {\mathcal P}^{{\rm cl}~j}_{{\rm grv}~l}     } {\bf z}_i^{\mathcal P},$$ and   $$ d^{{\rm cl}~j}_{{\rm grv}~l} = {\rm median}_{{\bf z}_i^{\mathcal P} \in {\mathcal P}^{{\rm cl}~j}_{{\rm grv}~l}}  \left \{  \left \| \left( {y}_{i}^{{\rm mv},~{\rm cl}~j},  {x}_{i}^{{\rm mv},~{\rm cl}~j}   \right)^T  -      \left(  y_{{{\rm grv}~l}}^{{\rm cl}~j}, x_{{\rm grv}~l}^{{\rm cl}~j} \right)^T   \right \| \right \}.$$

\end{tcolorbox}

Stationary sampling points are meant to describe the stationary points in the training set. For each stationary cluster, the number of stationary sampled points meant to summarize it are $N/\epsilon$, and they are randomly selected according to the following algorithm:

\begin{algorithm}[H]

	\SetKwInOut{Input}{input}
	\SetKwInOut{Output}{output}
	\SetKwFunction{QueryNeighborPoints}{QueryNeighborPoints}
	\SetKwFunction{StationarySamplingPoints}{StationarySamplingPoints}
	\SetKwFunction{add}{add}
	\SetKwFunction{MergeClusters}{MergeClusters}
	\SetKwFunction{remove}{remove}
	\SetKwFunction{distance}{distance}

\Input{${\mathcal P}^{{\rm st},~{\rm cl}~j}$,  $N_{\rm min}$, $\epsilon_{\rm Dist}$}
	\Output{${\mathcal S}^{{\rm cl}~j}$}

\BlankLine

	$ {\mathcal S}^{{\rm cl}~j} \leftarrow$ empty 

	${\rm Lat}_1,~{\rm Lat}_2 \leftarrow $ minimum and maximum of all points' latitude in ${\mathcal P}^{{\rm st},~{\rm cl}~l}$

	${\rm Lon}_1,~{\rm Lon}_2 \leftarrow $ minimum and maximum of all points' longitude in ${\mathcal P}^{{\rm st},~{\rm cl}~l}$

	$ {\rm Area} \leftarrow \left|  \left( {\rm Lat}_1 - {\rm Lat}_2 \right) \cdot \left( {\rm Lon}_1 - {\rm Lon}_2 \right) \right|$

	\If{\rm ${\rm Area} = 0$} {
		$sample\_size = 1$}
	\Else {
		$sample\_size = ceiling \left( { \left. {\rm Area} \right/ \left( \pi \cdot \epsilon_{\rm Dist}^2 \right) } \right)$}

	$count \leftarrow 0$ 

	\While{ $count < sample\_size$}{ 
		Randomly select one point from cluster ${\mathcal P}^{{\rm st},~{\rm cl}~l}$

		\If{ Randomly selected point is far from all points in ${\mathcal S}^{{\rm cl}~j}$} {

			Add point to ${\mathcal S}^{{\rm cl}~j}$

				$count = count + 1$}}

\caption{Extracting Stationary Sampling Points from ${\mathcal P}^{{\rm st,~cl}~j}$}
\end{algorithm}

Assume there are $n_{\rm ssp}^{{\rm cl}~j}$ stationary sampled points in stationary cluster $j$. We will call this set of points ${\mathcal S}^{{\rm cl}~j}$, and $${\mathcal S}^{{\rm cl}~j} = \left \{ {\bf s}_1^{{\rm cl}~j}, {\bf s}_2^{{\rm cl}~j}, \ldots, {\bf s}^{{\rm cl}~j}_{n^{{\rm cl}~j}_{\rm ssp}} \right \},$$ where $${\bf s}_l^{{\rm cl}~j} = \left( y^{{\rm cl}~j}_{{\rm ssp}~l},  x^{{\rm cl}~j}_{{\rm ssp}~l},  s^{{\rm cl}~j}_{{\rm ssp}~l},  c^{{\rm cl}~j}_{{\rm ssp}~l} \right).$$

It is with the gravity vectors and stationary points that an anomaly score is calculated. The subsection below explains how Liu (\cite{Liu2}) calculates this anomaly. We add some mathematical and statistical detail to their calculations and also introduce an alternative and more flexible  way to measure anomalous behavior.

\subsection{Calculating Anomalous Behavior}\label{sctn:CalcAnom}

To assign an anomaly score to ${\mathcal N}$, Liu et al. (\cite{Liu2}) first split the new track into its stationary and moving parts, ${\mathcal N}^{\rm st}$ and ${\mathcal N}^{\rm mv}.$   For each point in ${\mathcal N}^{\rm st}$, they calculate the smallest distance between it and the set of stationary sampled points.  This distance is called the Absolute Distance Deviation ($ADD$), and for point $i$ in ${\mathcal N}^{\rm st}$, it is calculated as $$ADD_i^{{\mathcal N}^{\rm st}} = \min_{j,l}  \left \{ \left \| \left( y_{i}^{{\mathcal N}^{\rm st}}, x_{i}^{{\mathcal N}^{\rm st}} \right)^T - \left( y_{{\rm ssp}~l}^{{\rm cl}~j},   x_{{\rm ssp}~l}^{{\rm cl}~j} \right)^T \right \| \right \}.$$

For each point in ${\mathcal N}^{\rm mv}$, they calculate two distance metrics, the Relative Distance Deviation (RDD) and the Cosine Division Distance (CDD). The RDD is similar to the ADD in that it calculates the smallest distance between a point and the set of gravity vectors, but this metric is different in that it accounts for the variation and geographical spread around the gravity vector. It does this  by dividing the distance by the median of the associated gravity vector.  For the $i^{\rm th}$ point in ${\mathcal N}^{\rm mv},$  $RDD$ is calculated as  $$RDD_i^{ {\mathcal N}^{\rm mv}} = \min_{l,j}  \left \{  \left.   \left \|  \left( y_{i}^{ {\mathcal N}^{\rm mv}}, x_{i}^{ {\mathcal N}^{\rm mv}} \right)^T -  \left( y_{{\rm grv}~l}^{{\rm cl}~j}, x_{{\rm grv}~l}^{{\rm cl}~j} \right)^T  \right \|   \right/ d_{{\rm grv}~l}^{{\rm cl}~j} \right \}. $$

The CDD accounts for any difference in heading and/or speed a point in ${\mathcal N}^{\rm mv}$ may have from the closest gravity point. For point $i$ in ${\mathcal N}^{\rm mv}$, it is calculated as   \begin{equation} \label{eqn:ourCDD} CDD_i^{{\mathcal N}^{\rm mv}} = {\rm cos(\alpha)} \cdot {\frac{ {\rm min} \left( s^*_{\rm grv}, s^{{\mathcal N}^{\rm mv}}_i \right)}{ {\rm max} \left( s^*_{\rm grv}, s^{{\mathcal N}^{\rm mv}}_i \right)}}, \end{equation} where $\alpha = \left| c^* - c_i^{{\mathcal N}^{\rm mv}} \right|$, and $c^*$ and $s^*$ are the course and speed components of gravity vector ${\bf g}^*$ (the closest gravity vector), where $${\bf g}^* = {\rm argmin}_{{\bf g} \in \cup_j {\mathcal G}^{{\rm cl}~j}}  \left \{  \left.   \left \|  \left( y_{i}^{ {\mathcal N}^{\rm mv}}, x_{i}^{ {\mathcal N}^{\rm mv}} \right)^T -  \left( y_{{\rm grv}~l}^{{\rm cl}~j}, x_{{\rm grv}~l}^{{\rm cl}~j} \right)^T  \right \|   \right/ d_{{\rm grv}~l}^{{\rm cl}~j} \right \}.   $$    The equation in (\ref{eqn:ourCDD}) looks at the ratio of the two different speeds and the magnitude in the difference of their directions. If the speeds and courses are identical, $CDD_i^{{\mathcal N}^{\rm mv}}  = 1,$ and the smaller its value, the greater the anomaly.

It should be noted that our definition of CDD is different than how it is written in Liu et al. (\cite{Liu2}).     They write CDD as $$  CDD_{{\rm Liu}~i}^{{\mathcal N}^{\rm mv}} = \max_{l,j}  \left \{ {\rm cos \left(   \left|   c_{{\rm grv}~l}^{{\rm cl}~j} - c_i^{{\mathcal N}^{\rm mv}}    \right|  \right)} \cdot {\frac{ {\rm min} \left( s^{{\rm cl}~j}_{{\rm grv}~l}, s^{{\mathcal N}^{\rm mv}}_i \right)}{ {\rm max} \left( s_{{\rm grv}~l}^{{\rm cl}~j}, s^{{\mathcal N}^{\rm mv}}_i \right)}} \right \}   ,$$ but this definition seems unclear. As they have written it, they are looking for a gravity vector  with a speed and course which most closely matches the speed and course of point $i$ in ${\mathcal N}^{\rm mv}$.  The location of this gravity vector is not considered. Given their definition of CDD, it would be possible for the point $i$ in ${\mathcal N}^{\rm mv}$ to be in the proximity of points going in an opposite direction and at a much different speed, yet still have a CDD value that went unnoticed. With this in mind, we modified/rewrote the definition of CDD.   For each point $i$ in ${\mathcal N}^{\rm mv}$, we evaluate $CDD$ at the closest gravity vector. The closest gravity vector is that which minimizes  RDD (${\bf g}^*$ in Equation \ref{eqn:ourCDD}).

With these three metrics (ADD, RDD and CDD) calculated, Liu et al. (\cite{Liu2}) score each observation in ${\mathcal N}$ depending on whether the calculated $ADD$, $RDD$, or $CDD$ are beyond a certain threshold.  For points in ${\mathcal N}^{\rm st}$, the score is calculated as $${\rm Scr}_{\rm Liu} \left( {\bf z}_i^{{\mathcal N}^{\rm st}} \right) = \left \{ \begin{array}{ll} 1 & {\rm if}~ADD_i^{{\mathcal N}^{\rm st}} > ADD_{\rm Threshold} \\ 0 & {\rm otherwise} \end{array} \right. .$$  For points in ${\mathcal N}^{\rm mv}$, the score is calculated as $${\rm Scr}_{\rm Liu} \left( {\bf z}_i^{{\mathcal N}^{\rm mv}}  \right) = \left \{ \begin{array}{ll} 1 & {\rm if}~RDD_i^{{\mathcal N}^{\rm mv}} > RDD_{\rm Threshold}~{\rm or}~CDD_i^{{\mathcal N}^{\rm mv}} < CDD_{\rm Threshold} \\ 0 & {\rm otherwise} \end{array} \right. $$ They then calculate the total anomaly score for the new trajectory, ${\mathcal N}$, as   $${\rm Anom}_{\rm Liu} \left( {\mathcal N} \right) = m^{-1} \left( \sum_{j=1}^{m^{\rm st}}  {\rm Scr}_{\rm Liu} \left( {\bf z}_j^{{\mathcal N}^{\rm st}} \right) + \sum_{j=1}^{m^{\rm mv}} {\rm Scr}_{\rm Liu} \left( {\bf z}_j^{{\mathcal N}^{\rm mv}}  \right)\right).$$

One way they obtain these three thresholds, $ADD_{\rm Threshold},$ $RDD_{\rm Threshold},$ and $CDD_{\rm Threshold}$ is by considering an entirely different data set, ${\mathcal D}$ (for ``D"ifferent), calculating the distribution of $ADD,$ $RDD$ and $CDD$ values in this data set, and then letting $ADD_{\rm Threshold}$ and $RDD_{\rm Threshold}$ be the $95^{\rm th}$ percentile of the distribution in the $ADD$ and $RDD$ values, and $CDD_{\rm Threshold}$ be the $5^{\rm th}$ percentile of the distribution in the $CDD$ values.  This is mathematically formulated in the textbox below.

\begin{tcolorbox}

\begin{center}

{\bf Calculating Threshold Values}

\end{center}

We assume ${\mathcal D}$ has $r$ observations, $r^{\rm st}$ which are stationary and $r^{\rm mv}$ which are moving. From this data set, we calculate $r^{\rm st}$ values of $ADD$, written as $$ {\mathcal {ADD}}^{\mathcal D} = \left \{  ADD_i^{{\mathcal D}^{\rm st}}  : 1 \leq i \leq r^{\rm st} \right \},$$ and $r^{\rm ~mv}$ values of $RDD$ and $CDD$, written as $$ {\mathcal {RDD}}^{\mathcal D} = \left \{   RDD_j^{{\mathcal D}^{\rm mv}}  : 1 \leq j \leq r^{\rm mv} \right \},$$ and $$ {\mathcal {CDD}}^{\mathcal D} = \left \{   CDD_j^{{\mathcal D}^{\rm mv}}  : 1 \leq j \leq r^{\rm mv} \right \}.$$   $ADD_{\rm Threshold}$ is the $95^{\rm th}$ percentile of ${\mathcal ADD}^{\mathcal D}$, $RDD_{\rm Threshold}$ is the $95^{\rm th}$ percentile of ${\mathcal RDD}^{\mathcal D}$, and $CDD_{\rm Threshold}$ is the $5^{\rm th}$ percentile of ${\mathcal CDD}^{\mathcal D}$.

\end{tcolorbox}

A setback to the method that Liu et al. use to measure anomalous behavior is that it fails to highlight the extremity of the anomaly.   For example, if ${\mathcal N}$ were a set of stationary points, all of which were just barely beneath $ADD_{\rm Threshold}$, its anomaly score would be 0.  Its anomaly score would also be 0 if all these stationary points were significantly below $ADD_{\rm Threshold}$. Yet another setback to Liu et al.'s  method  is that the statistical significance of their anomaly   is not immediately transparent. For instance, if ${\rm Anom}_{\rm Liu} \left( {\mathcal N} \right) = 0.3$, it is not obvious from this statistic what the probability is of observing an anomaly  as or more extreme than the one observed.

 Part of the novelty proposed in this paper is in how we calculate the anomaly of a new trajectory. We calculate the anomaly in such a way that the two extreme cases described above would have considerably  different scores.   The statistical significance of our anomaly score  is also transparent (we ultimately report a $z-$score).    

We begin by scoring each observation not with a 1 or a 0 (as Liu et al. did), but with  the fraction of ${\mathcal ADD}^{\mathcal D},~{\mathcal RDD}^{\mathcal D},$  and ${\mathcal CDD}^{\mathcal D}$ values that are more extreme than the one observed.  The smaller this fraction, the more unusual/extreme an observation. These scores are written below in Equations \ref{eqn:Score_St} and \ref{eqn:Score_Mv}.  \begin{eqnarray}  \label{eqn:Score_St}   {\rm Scr}_{\rm Botts} \left( {\bf z}_i^{{\mathcal N}^{\rm st}} \right) &  =  &   {\frac{1}{r^{\rm st}}} \sum_{j=1}^{r^{\rm st}}  {\mathbbm{1}} \left( ADD_j^{ {\mathcal D}^{\rm st}} \geq ADD_i^{ {\mathcal N}^{\rm st}} \right),~~{\rm and} \\   \label{eqn:Score_Mv}  {\rm Scr}_{\rm Botts} \left( {\bf z}_i^{{\mathcal N}^{\rm mv}} \right) &  = &    \min \left[   {\frac{1}{r^{\rm mv}}}\sum_{j=1}^{r^{\rm mv}}     {\mathbbm{1}} \left( RDD_j^{{\mathcal D}^{\rm mv}} \geq RDD_i^{{\mathcal N}^{\rm mv}}  \right)  \right. ,  \\   \nonumber & & \left.  {\frac{1}{r^{\rm mv}}}   \sum_{j=1}^{r^{\rm mv}}   {\mathbbm{1}} \left( CDD_j^{{\mathcal D}^{\rm mv}} \leq CDD_i^{{\mathcal N}^{\rm mv}}  \right) \right].  \end{eqnarray}    The quantity ${\rm Scr}_{\rm Botts} \left(  {\bf z}_i^{{\mathcal N}^{\rm st}} \right)$ estimates the probability that any randomly selected ADD value will be more extreme than the one observed.   The quantity ${\rm Scr}_{\rm Botts} \left( {\bf z}_i^{{\mathcal N}^{\rm mv}} \right)$ considers the probability of observing an RDD greater than the one observed and a CDD less than the one observed, and returns the smaller of the two.  The distribution of ${\rm Scr}_{\rm Botts} \left( {\bf z}_i^{{\mathcal N}^{\rm st}} \right)$ can be approximated with that of a uniform random variable, $U_1$, where ${\mathbbm E} \left( U_1 \right) = 0.5$ and ${\rm Var} \left( U_1 \right) = 1/12.$    The distribution of ${\rm Scr}_{\rm Botts} \left( {\bf z}_i^{{\mathcal N}^{\rm mv}} \right)$ can be approximated with that of the minimum of two uniform random variables, $U_2$ and $U_3$.    If $U_{\rm min} = \min \left( U_2, U_3 \right)$, then ${\mathbbm E} \left( U_{\rm min} \right) = 1/3$ and ${\rm Var} \left( U_{\rm min} \right) = 1/18.$

With ${\rm Scr}_{\rm Botts} \left( {\bf z}_i^{{\mathcal N}^{\rm st}} \right)$ and ${\rm Scr}_{\rm Botts} \left( {\bf z}_i^{{\mathcal N}^{\rm mv}} \right)$ approximated by these distributions that have known first and second moments, we can apply the central limit theorem and conclude that \begin{equation} \label{eqn:W_st}   W^{\rm st}  =   \left \{   \left. \left[  \left(    m^{\rm st} \right)^{-1}    \sum_{j=1}^{m^{\rm st}} {\rm Scr}_{\rm Botts} \left( {\bf z}_j^{{\mathcal N}^{\rm st}} \right) - .5 \right]  \right/ \sqrt{ {\frac{1}{12 m^{\rm st}}} } \right \}     \stackrel{d}{\longrightarrow} N(0,1)~{\rm as}~r^{\rm st}~\&~m^{\rm st} \rightarrow \infty, \end{equation}     and   \begin{equation}  \label{eqn:W_mv}     W^{\rm mv}  =   \left \{   \left. \left[  \left(    m^{\rm mv} \right)^{-1}    \sum_{j=1}^{m^{\rm mv}} {\rm Scr}_{\rm Botts} \left( {\bf z}_j^{{\mathcal N}^{\rm mv}} \right) -  {\frac{1}{3}}  \right]  \right/ \sqrt{ {\frac{1}{18 m^{\rm mv}}} } \right \}     \stackrel{d}{\longrightarrow} N(0,1)~{\rm as}~r^{\rm mv}~\&~m^{\rm mv} \rightarrow \infty .   \end{equation}

Our final anomaly statistic, ${\rm Anom}_{\rm Botts} \left( {\mathcal N} \right)$, combines the two asymptotically normal random variables in Equations \ref{eqn:W_st} and \ref{eqn:W_mv} as shown below in Equation \ref{eqn:BottsAnom}.

\begin{equation}   \label{eqn:BottsAnom}  {\rm Anom}_{\rm Botts} \left( {\mathcal N} \right) = \left \{    \begin{array}{ll} W^{\rm st} & {\rm if}~ m^{\rm mv} = 0 \\ W^{\rm mv} & {\rm if}~m^{\rm st} = 0 \\  \left. \left( W^{\rm st} + W^{\rm mv} \right)   \right/ \sqrt{2} & {\rm if}~m^{\rm st} > 0 ~\& ~m^{\rm mv} > 0  \end{array}  \right.   . \end{equation}

Assuming independence across observations in the track ${\mathcal N}$, and assuming independence of the variables $ADD, RDD,$ and $CDD$, the expected value and variance of  ${\rm Anom}_{\rm Liu} \left( {\mathcal N} \right)$ are  \begin{eqnarray*}    {{\mathbbm E}} \left[  {\rm Anom}_{\rm Liu} \left( {\mathcal N} \right) \right] &  = &  .05 + \left( \left. m^{\rm mv} \right/ m \right) \cdot \left( .05 - .05^2 \right),~~~~{\rm and} \\ {\rm Var} \left[ {\rm Anom}_{\rm Liu} \left( {\mathcal N} \right) \right] & = &      \left( m^{\rm st} + m^{\rm mv} \right)^{-2} \left[ m^{\rm st} \cdot .05 \cdot .95 + 2  m^{\rm mv} \cdot .05 \cdot .95 +    m^{\rm mv} \cdot .05^2 \cdot  \left( 1 - .05^2 \right) \right], \end{eqnarray*}  and the asymptotic distribution of ${\rm Anom}_{\rm Botts} \left( {\mathcal N} \right)$ is $N(0,1)$,    making $${\mathbbm E} \left[ {\rm Anom}_{\rm Botts} \left( {\mathcal N} \right) \right]= 0~~~{\rm and}~~~{\rm Var} \left[ {\rm Anom}_{\rm Botts} \left( {\mathcal N} \right) \right] = 1.$$    Theorem \ref{thm:LiuExpectation} in Section \ref{sctn:Proofs} of the Appendix calculates the expected value and variance of ${\rm Anom}_{\rm Liu} \left( {\mathcal N} \right)$. Theorem \ref{thm:BottsDistn} in Section \ref{sctn:Proofs} of the Appendix justifies the asymptotic normality of ${\rm Anom}_{\rm Botts} \left( {\mathcal N} \right)$. In Section \ref{sctn:Simulations} of the Appendix, some simulation results are provided which illustrate this asymptotic normality and its invariance to the distribution of $ADD$, $CDD$, and $RDD$.

We would expect normal/in-family trajectories to have ${\rm Anom}_{\rm Liu} \left( {\mathcal N} \right)$ and ${\rm Anom}_{\rm Botts} \left( {\mathcal N} \right)$ values close to these expected values. For anomalous trajectories, we would expect  ${\rm Anom}_{\rm Liu} \left( {\mathcal N} \right)$ to be  considerably larger than its expected value, and ${\rm Anom}_{\rm Botts} \left( {\mathcal N} \right)$ to be  considerably less than its expected value. Remember that ${\rm Anom}_{\rm Liu} \left( {\mathcal N} \right)$ counts the fraction of points in ${\mathcal N}$ that are beyond a certain threshold.  The larger the value of ${\rm Anom}_{\rm Liu} \left( {\mathcal N} \right)$, the more anomalous a trajectory is.  ${\rm Anom}_{\rm Botts} \left( {\mathcal N} \right)$ considers the fraction of values in ${\mathcal{ADD}}^{\mathcal D},~{\mathcal{RDD}}^{\mathcal D},~{\rm and}~{\mathcal{CDD}}^{\mathcal D}$ that are more extreme than those observed in ${\mathcal N}$ and turns this into a $z-$score. The smaller the value of ${\rm Anom}_{\rm Botts} \left( {\mathcal N} \right)$, the more anomalous a trajectory. It is thus expected that ${\rm Anom}_{\rm Liu} \left( {\mathcal N} \right)$ and ${\rm Anom}_{\rm Botts} \left( {\mathcal N} \right)$ will jointly go in opposite directions of their respective expectations.

Figures \ref{fig:Example1} and \ref{fig:Example2} show normal trajectories with values of ${\rm Anom}_{\rm Liu} \left( {\mathcal N} \right)$ and ${\rm Anom}_{\rm Botts} \left( {\mathcal N} \right)$ that are close to their expected values.   Figures \ref{fig:Example3} and \ref{fig:Example4} show abnormal trajectories. In these cases, ${\rm Anom}_{\rm Liu} \left( {\mathcal N} \right)$ is  higher than expected, and ${\rm Anom}_{\rm Botts} \left( {\mathcal N} \right)$ is significantly lower than expected. Assuming the $ADD$, $RDD$, and $CDD$ values in ${\mathcal N}$ come from the same distribution as those in the training set and ${\mathcal D}$, the probability of observing ${\rm Anom}_{\rm Botts} \left( {\mathcal N} \right)$ scores as or more extreme than the ones observed in Figures \ref{fig:Example3} and \ref{fig:Example4} is less than $10^{-68}$.

Figures \ref{fig:Example5} and \ref{fig:Example6} show other sets of ${\mathcal N}$ with surprising results of ${\rm Anom}_{\rm Liu} \left( {\mathcal N} \right)$ and ${\rm Anom}_{\rm Botts} \left( {\mathcal N} \right).$ In both cases, the values of ${\rm Anom}_{\rm Liu} \left( {\mathcal N} \right)$  do not suggest anomalous behavior, yet the values of ${\rm Anom}_{\rm Botts} \left( {\mathcal N} \right)$  do.   They are  less than expected for a normal trajectory. The probabilities of observing ${\rm Anom}_{\rm Botts} \left( {\mathcal N} \right)$ scores as low as those observed in Figures \ref{fig:Example5} and \ref{fig:Example6} are lower than $2 \times 10^{-15}$ and $2 \times 10^{-4}$, respectively.      In both of these figures, the difference between the two metrics  illustrates why ${\rm Anom}_{\rm Botts} \left( {\mathcal N} \right)$ may be a preferred metric to ${\rm Anom}_{\rm Liu} \left( {\mathcal N} \right).$

In Figure \ref{fig:Example5}, ${\mathcal N}$ is made up entirely of stationary points, and they are all (nearly) at the same location. Each point has an $ADD$ that is approximately .017, which is beneath the $ADD$ threshold of .034. Since they are all beneath the threshold, ${\rm Anom}_{\rm Liu} \left( {\mathcal N} \right) = 0$. ${\rm Anom}_{\rm Botts} \left( {\mathcal N} \right)$, however, does capture how extreme these values of $ADD$ are, since it does not depend on a threshold and merely counts the number of $ADD^{\mathcal D}$ values that are greater than it.

The same principle applies in Figure \ref{fig:Example6}. In this case, very few $CDD$ values in the moving part of ${\mathcal N}$ are below the $CDD$ threshold of -.746,  and no $RDD$ values are above the $RDD$ threshold of 2.255,  yet 62\% of ${\mathcal N}$'s $RDD$ values are within the $95^{\rm th}$ and $70^{\rm th}$ percentile of the $RDD$ values in ${\mathcal RDD}^{\mathcal D}$.

\section{Conclusion}

In this paper, we illustrate, correct, and clarify the DBSCANSD clustering algorithm presented in Liu et al. (\cite{Liu1} ). The DBSCANSD algorithm is one of few algorithms that allows one to cluster AIS data based on location, speed, and course. With these clustering results, we offer an alternative to the anomaly metric presented in Liu et al. (\cite{Liu2}). The statistical significance of the metric we propose is transparent (unlike the one in Liu et al. (\cite{Liu2})), and the asymptotic distribution of our proposed statistic is invariant to the distributions of course, speed, or location in the training data.  In the future we hope to consider other variables such as heading and/or time in the clustering algorithm. Perhaps they can also be considered when trying to detect anomalous behavior in ships.

{
\begin{figure}[H] 
\begin{center}
{\includegraphics[trim = 5cm 6cm 5cm 6cm, clip, scale = 0.25]{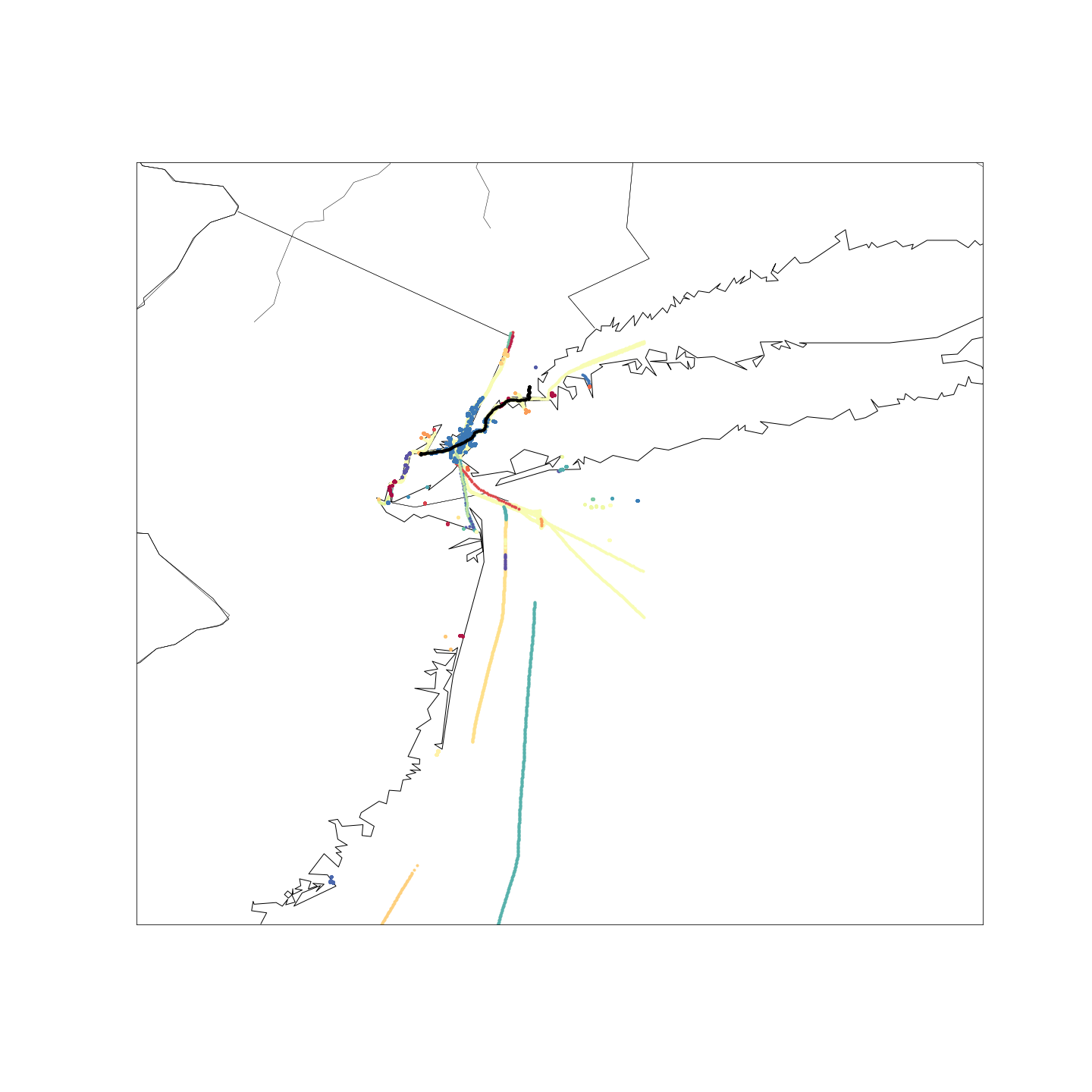}}

\caption{The values of ${\mathcal N}$ are in black, $m^{\rm st} = 69,$ and $m^{\rm mv} = 176.$ $ {\mathbbm E}  \left[ {\rm Anom}_{\rm Liu} \left( {\mathcal N} \right) \right] = 0.0841,~ {\rm StDev} \left[    {\rm Anom}_{\rm Liu} \left( {\mathcal N} \right)     \right] = .0184,~$  $ {\mathbbm E}  \left[ {\rm Anom}_{\rm Botts} \left( {\mathcal N} \right) \right] = 0,~{\rm and}~ {\rm StDev} \left[    {\rm Anom}_{\rm Botts} \left( {\mathcal N} \right)     \right] = 1.$      ${\rm Anom}_{\rm Liu} \left( {\mathcal N} \right) = 0.094,~{\rm and}~{\rm Anom}_{\rm Botts} \left(  {\mathcal N} \right) = 0.767.$ These numbers indicate no abnormality, and the picture illustrates a common ship path north of Long Island, NY.}

\label{fig:Example1}
\end{center}
\end{figure}
}

{
\begin{figure}[H] 
\begin{center}
{\includegraphics[trim = 6cm 12cm 5cm 7cm, clip, scale = 0.25]{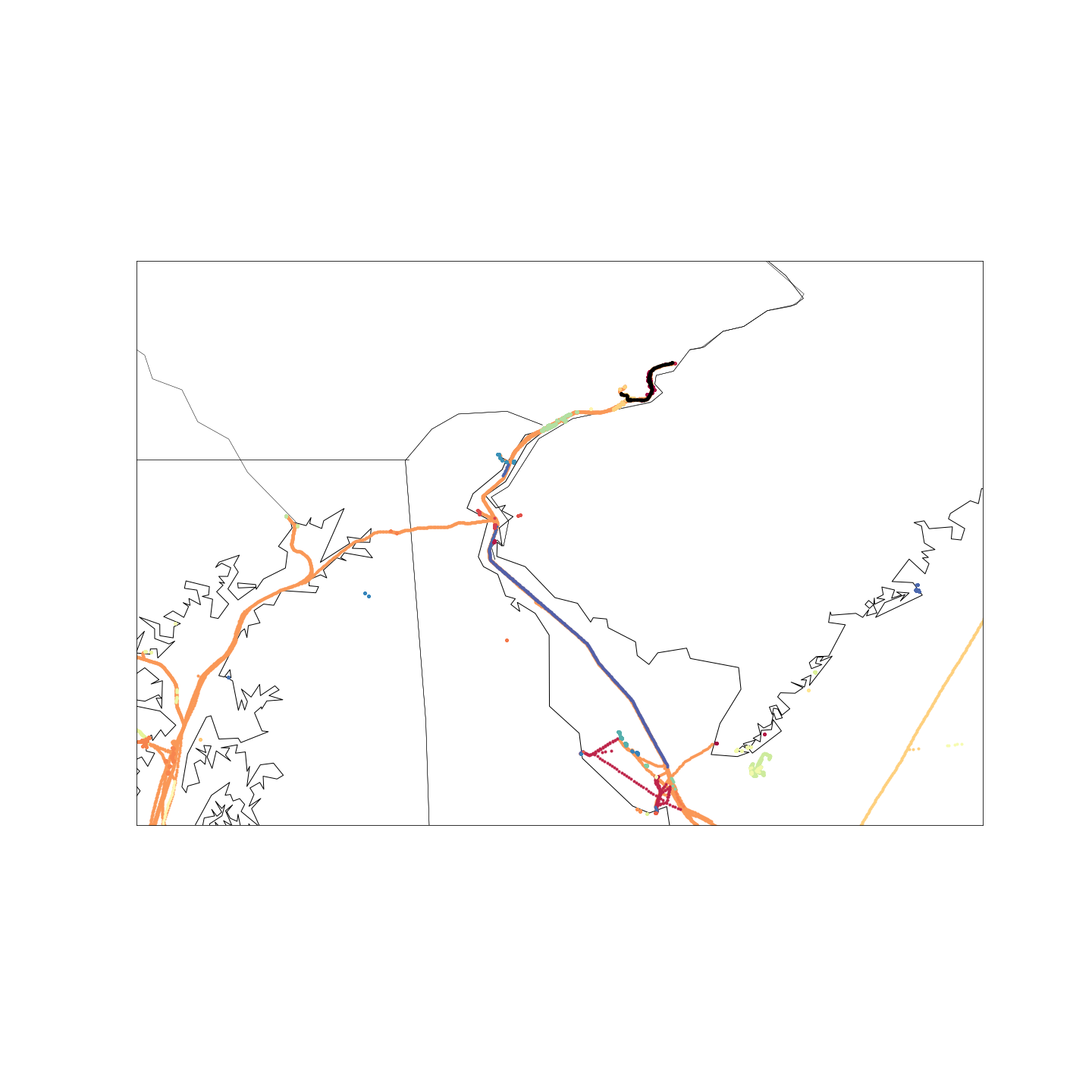}}

\caption{The values of ${\mathcal N}$ are in black, $m^{\rm st} = 31$, and $m^{\rm mv} = 172.$ $ {\mathbbm E}  \left[ {\rm Anom}_{\rm Liu} \left( {\mathcal N} \right) \right] = 0.0902,~ {\rm StDev} \left[    {\rm Anom}_{\rm Liu} \left( {\mathcal N} \right)     \right] = .0210,$  $~ {\mathbbm E}  \left[ {\rm Anom}_{\rm Botts} \left( {\mathcal N} \right) \right] = 0,~{\rm and}~ {\rm StDev} \left[    {\rm Anom}_{\rm Botts} \left( {\mathcal N} \right)     \right] = 1.$      ${\rm Anom}_{\rm Liu} \left( {\mathcal X}^{\mathcal P} \right) = 0.064,~{\rm and}~{\rm Anom}_{\rm Botts} \left(  {\mathcal N} \right) = -2.564.$ These numbers indicate no severe abnormality, and the picture illustrates a common ship path in the Delaware River.}

\label{fig:Example2}
\end{center}
\end{figure}
}

{
\begin{figure}[H] 
\begin{center}
{\includegraphics[trim = 0cm 6.5cm 0cm 6.5cm, clip, scale = 0.25]{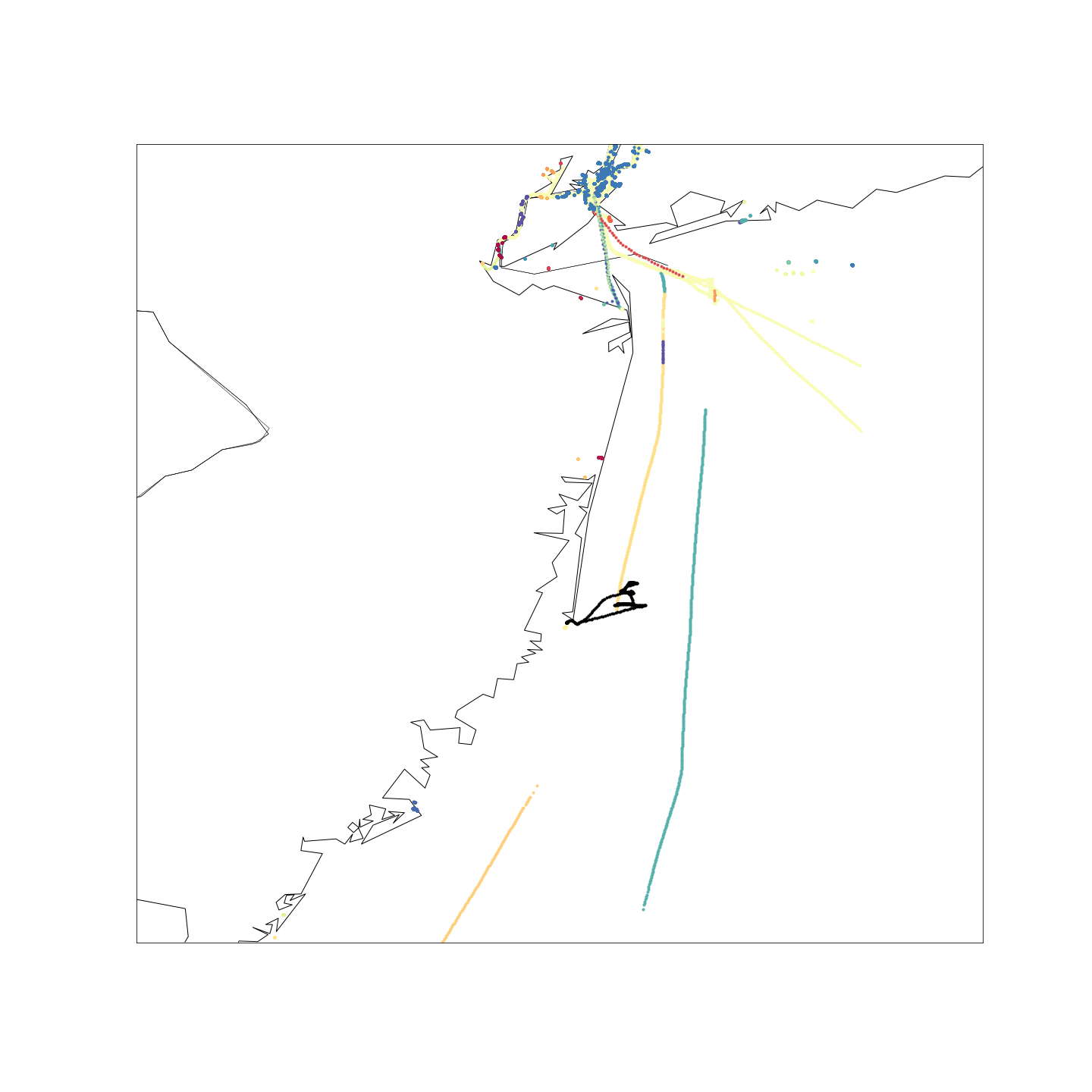}}

\caption{The values of  ${\mathcal N}$ are in black, $m^{\rm st} = 37$, and $m^{\rm mv} = 249.$ $ {\mathbbm E}  \left[ {\rm Anom}_{\rm Liu} \left( {\mathcal N} \right) \right] = 0.0914,~ {\rm StDev} \left[    {\rm Anom}_{\rm Liu} \left( {\mathcal N} \right)     \right] = .0178$, $ {\mathbbm E}  \left[ {\rm Anom}_{\rm Botts} \left( {\mathcal N} \right) \right] = 0,~ {\rm and}~{\rm StDev} \left[    {\rm Anom}_{\rm Botts} \left( {\mathcal N} \right)     \right] = 1.$      ${\rm Anom}_{\rm Liu} \left( {\mathcal N}^{\mathcal P} \right) = 0.706,~{\rm and}~{\rm Anom}_{\rm Botts} \left(  {\mathcal N} \right) = -17.616.$ These numbers suggest abnormal behavior, and the picture illustrates abnormal behavior off the coast of New Jersey.}

\label{fig:Example3}
\end{center}
\end{figure}
}

{
\begin{figure}[H] 
\begin{center}
{\includegraphics[trim = 6cm 8cm 5cm 8cm, scale = 0.25]{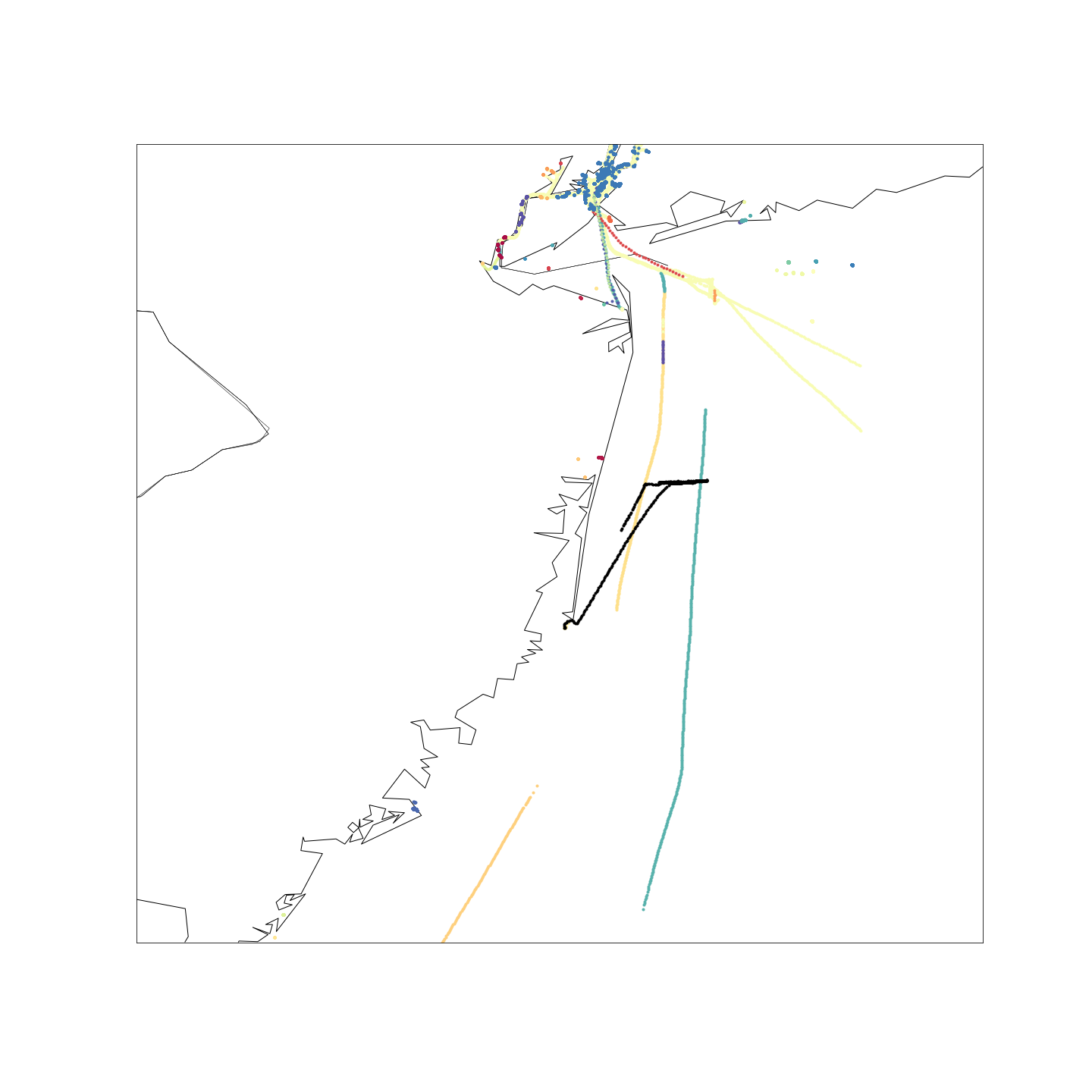}}

\caption{The values of ${\mathcal N}$ are in black, $m^{\rm st} = 39$, and $m^{\rm mv} = 379.$ $ {\mathbbm E}  \left[ {\rm Anom}_{\rm Liu} \left( {\mathcal N} \right) \right] = 0.0931,~ {\rm StDev} \left[    {\rm Anom}_{\rm Liu} \left( {\mathcal N} \right)     \right] = .0149$, $ {\mathbbm E}  \left[ {\rm Anom}_{\rm Botts} \left( {\mathcal N} \right) \right] = 0,~{\rm and}~ {\rm StDev} \left[    {\rm Anom}_{\rm Botts} \left( {\mathcal N} \right)     \right] = 1.$      ${\rm Anom}_{\rm Liu} \left( {\mathcal N} \right) = 0.770,~{\rm and}~{\rm Anom}_{\rm Botts} \left(  {\mathcal N} \right) = -23.776.$ These numbers suggest abnormal behavior, and the picture illustrates abnormal behavior off the coast of New Jersey.}

\label{fig:Example4}
\end{center}
\end{figure}
}

{
\begin{figure}[H] 
\begin{center}
{\includegraphics[trim = 7cm 7cm 7cm 7cm, scale = 0.25]{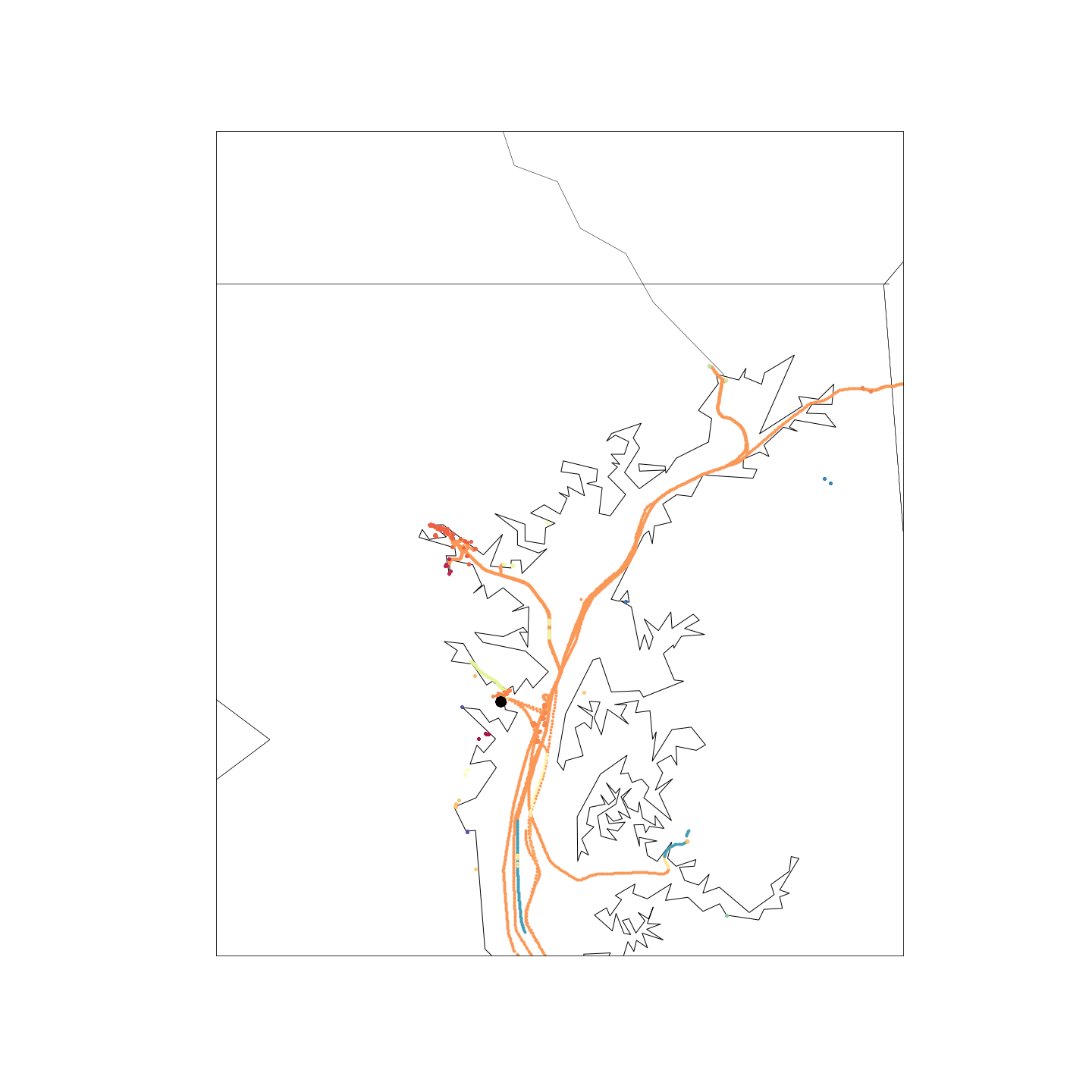}}

\caption{The values of ${\mathcal N}$ are in black, $m^{\rm st} = 34$, and $m^{\rm mv} = 0.$ $ {\mathbbm E}  \left[ {\rm Anom}_{\rm Liu} \left( {\mathcal N} \right) \right] = 0.05,  {\rm StDev} \left[    {\rm Anom}_{\rm Liu} \left( {\mathcal N} \right)     \right] = .0374$, $ {\mathbbm E}  \left[ {\rm Anom}_{\rm Botts} \left( {\mathcal N} \right) \right] = 0,~{\rm and}~ {\rm StDev} \left[    {\rm Anom}_{\rm Botts} \left( {\mathcal N} \right)     \right] = 1.$      ${\rm Anom}_{\rm Liu} \left( {\mathcal N} \right) = 0.0,~{\rm and}~{\rm Anom}_{\rm Botts} \left(  {\mathcal N} \right) = -8.191.$    ${\rm Anom}_{\rm Liu}$ suggests no abnormal behavior, yet ${\rm Anom}_{\rm Botts}$ does. In this case, the stationary points in ${\mathcal N}$ are just below the threshold used in calculating ${\rm Anom}_{\rm Liu}$.}

\label{fig:Example5}
\end{center}
\end{figure}
}

{
\begin{figure}[H] 
\begin{center}
{\includegraphics[trim = 7cm 7cm 7cm 7cm, scale = 0.25]{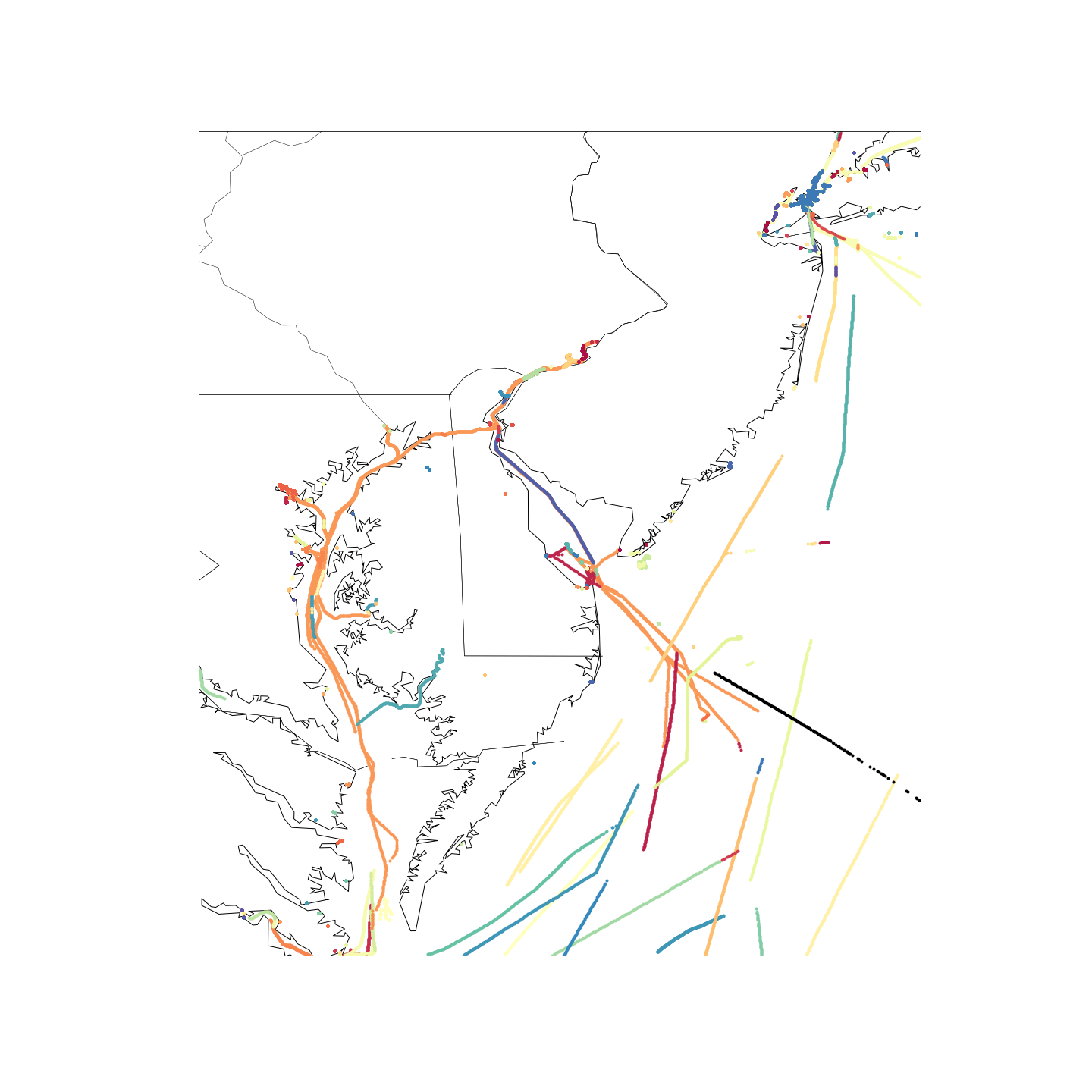}}

\caption{The values of ${\mathcal N}$ are in black, $m^{\rm st} = 0$, and $m^{\rm mv} = 160.$ $ {\mathbbm E}  \left[ {\rm Anom}_{\rm Liu} \left( {\mathcal N} \right) \right] = 0.0975, {\rm StDev} \left[    {\rm Anom}_{\rm Liu} \left( {\mathcal N} \right)     \right] = .0247$, $ {\mathbbm E}  \left[ {\rm Anom}_{\rm Botts} \left( {\mathcal N} \right) \right] = 0,~{\rm and}~ {\rm StDev} \left[    {\rm Anom}_{\rm Botts} \left( {\mathcal N} \right)     \right] = 1.$      ${\rm Anom}_{\rm Liu} \left( {\mathcal N} \right) = 0.063,~{\rm and}~{\rm Anom}_{\rm Botts} \left(  {\mathcal N} \right) = -3.640.$ In this case, ${\rm Anom}_{\rm Liu}$ suggests no abnormal behavior, yet ${\rm Anom}_{\rm Botts}$ does. Nearly all of the moving points  in ${\mathcal N}$ have $RDD$ values just below the threshold used in calculating ${\rm Anom}_{\rm Liu}$}

\label{fig:Example6}
\end{center}
\end{figure}
}


\appendix

\begin{center}

\begin{Large}

{\bf Appendix}

\end{Large}

\end{center}

\section{Proof of Theorems}\label{sctn:Proofs} 

\begin{thm} \label{thm:LiuExpectation}  Consider six sets of random variables, ${\mathcal V}^{\mathcal D},$ ${\mathcal V}^{\mathcal N},$ ${\mathcal Y}^{\mathcal D},$ ${\mathcal Y}^{\mathcal N},$ ${\mathcal Z}^{\mathcal D}$, and ${\mathcal Z}^{\mathcal N}$ where \begin{eqnarray*}    {\mathcal V}^{\mathcal D} & = & \left \{ V_1^{\mathcal D}, V_2^{\mathcal D}, \ldots, V_{r^{\rm st}}^{\mathcal D} \right \} ~~~~~~~{\mathcal V}^{\mathcal N}  =  \left \{ V_1^{\mathcal N}, V_2^{\mathcal N}, \ldots, V_{m^{\rm st}}^{\mathcal N} \right \} \\ {\mathcal Y}^{\mathcal D} & = & \left \{ Y_1^{\mathcal D}, Y_2^{\mathcal D}, \ldots, Y_{r^{\rm mv}}^{\mathcal D} \right \} ~~~~~~~ {\mathcal Y}^{\mathcal N}  =  \left \{ Y_1^{\mathcal N}, Y_2^{\mathcal N}, \ldots, Y_{m^{\rm mv}}^{\mathcal N} \right \}\\ {\mathcal Z}^{\mathcal D} & = & \left \{ Z_1^{\mathcal D}, Z_2^{\mathcal D}, \ldots, Z_{r^{\rm mv}}^{\mathcal D} \right \}~~~~~~~{\mathcal Z}^{\mathcal N} = \left \{ Z_1^{\mathcal D}, Z_2^{\mathcal D}, \ldots, Z_{r^{\rm mv}}^{\mathcal D} \right \}. \end{eqnarray*}   Assume these sets of random variables are independent of one another and also assume that the variables within each set are independent of another.    We will also assume that all values in ${\mathcal V}^{\mathcal D}$ and ${\mathcal V}^{\mathcal N}$ follow the common distribution $f_V(v)$, all the values in ${\mathcal Y}^{\mathcal D}$ and ${\mathcal Y}^{\mathcal N}$ follow the common distribution $f_Y(y)$, and all the values in ${\mathcal Z}^{\mathcal D}$ and ${\mathcal Z}^{\mathcal N}$ follow the common distribution $f_Z(z)$.     With these definitions, consider the statistic $T_1(\alpha)$, where: \footnote{In this theorem, $V_1^{\mathcal D}, V_2^{\mathcal D}, \ldots, V_{r^{\rm st}}^{\mathcal D}$ represent the $r^{\rm st}$ values of ${\mathcal ADD}^{\mathcal D}$,  $V_1^{\mathcal N}, V_2^{\mathcal N}, \ldots, V_{m^{\rm st}}^{\mathcal N}$ represent the $m^{\rm st}$ values of ADD in trajectory ${\mathcal N}, $  $Y_1^{\mathcal D}, Y_2^{\mathcal D}, \ldots, Y_{r^{\rm st}}^{\mathcal D}$ represent the $r^{\rm mv}$ values of ${\mathcal RDD}^{\mathcal D}$,  $Y_1^{\mathcal N}, Y_2^{\mathcal N}, \ldots, Y_{m^{\rm mv}}^{\mathcal N}$ represent the $m^{\rm mv}$ values of RDD in trajectory ${\mathcal N}, $   $Z_1^{\mathcal D}, Z_2^{\mathcal D}, \ldots, Z_{r^{\rm mv}}^{\mathcal D}$ represent the $r^{\rm mv}$ values of ${\mathcal CDD}^{\mathcal D}$,  and $Z_1^{\mathcal N}, Z_2^{\mathcal N}, \ldots, Z_{m^{\rm mv}}^{\mathcal N}$ represent the $m^{\rm mv}$ values of CDD in trajectory ${\mathcal N}. $   The statistic $T_1$ takes the form of ${\rm Anom}_{\rm Liu} \left( {\mathcal N} \right).$}     \begin{eqnarray*}   T_1(\alpha) & = & {\frac{1}{m^{\rm Tot}}}  \left \{   \sum_{j=1}^{m^{\rm st}}  {\mathbbm 1} \left( V_j^{\mathcal N} \geq {\hat {Q}}_{V, r^{\rm st}}(1-\alpha) \right)       + \sum_{j=1}^{m^{\rm mv}} {\mathbbm 1}  \left( Y_j^{\mathcal N} \geq {\hat {Q}}_{Y, r^{\rm mv}}(1 - \alpha)  ~{\rm or}~ Z_j^{\mathcal N} \leq {\hat {Q}}_{Z, r^{\rm mv}} ( \alpha ) \right)  \right \}           \end{eqnarray*}  where  $m^{\rm Tot} = m^{\rm st} + m^{\rm mv},$ and ${\hat Q}_{V, r^{\rm st}}(\gamma)$ is the $\gamma^{\rm th}$ percentile of ${\mathcal V}^{\mathcal D}$.    Then  \begin{eqnarray*}  {\mathbbm E} \left[ T_1  \left( \alpha \right) \right]  & \longrightarrow &  \alpha + \left. \left[ m^{\rm mv} \right/  \left( m^{\rm st} + m^{\rm mv} \right) \right] \cdot  \left( \alpha - \alpha^2 \right) ~{\rm as}~ \min \left( r^{\rm mv}, r^{\rm st} \right) \longrightarrow \infty,~~~{\rm and} \\ {\rm Var} \left[ T_1 (\alpha) \right] & \longrightarrow & \left( m^{\rm Tot} \right)^{-2} \left[  m^{\rm st}   \alpha (1 - \alpha ) + 2  m^{\rm mv}  \alpha  (1 - \alpha) +   m^{\rm mv}  \alpha^2  (1- \alpha^2) \right],~{\rm as}~ \min \left( r^{\rm st}, r^{\rm mv} \right) \longrightarrow \infty. \end{eqnarray*}    \end{thm}

\begin{proof} First observe that ${\hat Q}_{V, r^{\rm st}} \left( \gamma \right) \stackrel{p}{\longrightarrow} Q_V(\gamma)$ as $r^{\rm st} \longrightarrow \infty,$ where $Q_V(\gamma)$ is the number such that ${\mathbbm P} \left( V \leq Q_V(\gamma) \right) = \gamma.$  (see Serfling, \cite{Serfling}). This implies ${\mathbbm 1} \left( V_j^{\mathcal N} \geq {\hat Q}_{V, r^{\rm st}} (\gamma) \right) \stackrel{d}{\longrightarrow} {\mathbbm 1} \left( V_j^{\mathcal N} \geq Q_V(\gamma) \right)$ as $r^{\rm st} \longrightarrow \infty. $ Since $$\sup_{r^{\rm st} \geq 1} \left \{ {\mathbbm E} \left[ \left| {\mathbbm 1} \left( V_j^{\mathcal N} \geq {\hat Q}_{V, r^{\rm st}} \left( \gamma \right) \right) \right|^l \right] \right \} < \infty~~~{\rm for~any~}l, $$ ${\mathbbm 1} \left( V_j^{\mathcal N} \geq Q_{V, r^{\rm st}} (\gamma) \right)$ is uniformly integrable (see Athreya \& Lahiri \cite{Athreya}). With uniform integrability, we can apply expectations to get   \begin{eqnarray*}     {\mathbbm E} \left[ T_1 (\alpha) \right] & \longrightarrow & \left( m^{\rm Tot} \right)^{-1} \left \{  \sum_{j=1}^{m^{\rm st}} {\mathbbm E} \left[ {\mathbbm 1} \left( V_j \geq Q_V( 1 - \alpha ) \right) \right] +  \sum_{j=1}^{m^{\rm mv}}  {\mathbbm E} \left( {\mathbbm 1}  \left[ Y_j \geq Q_Y \left( 1 - \alpha  \right)~{\rm or}~Z_j \leq Q_Z \left( \alpha \right) \right) \right]  \right \} \\     & = & \left( m^{\rm Tot} \right)^{-1} \left \{  \sum_{j=1}^{m^{\rm st}} {\mathbbm P} \left[  V_j \geq Q_V( 1 - \alpha ) \right]  +  \sum_{j=1}^{m^{\rm mv}}  {\mathbbm P} \left[  \left( Y_j \geq Q_Y \left( 1 - \alpha  \right) \right)~{\rm or}~ \left( Z_j \leq Q_Z \left( \alpha \right) \right) \right]  \right \} \\    & = & \left( m^{\rm Tot} \right)^{-1} \left \{  m^{\rm st} \cdot \alpha   + m^{\rm mv} \cdot \left[ {\mathbb P} ( Y \geq Q_Y \left( 1 - \alpha \right) ) + {\mathbb P} \left( Z_j \leq Q_Z \left( \alpha \right) \right)    \right. \right.  \\ & & \left. \left. -   {\mathbb P} \left( Y \geq Q_Y( 1 - \alpha )~{\rm and} ~Z_j \leq Q_Z( \alpha ) \right) \right]   \right \}  \\ & = & \left( m^{\rm Tot} \right)^{-1} \left[  m^{\rm st} \cdot \alpha + m^{\rm mv} \cdot \left( \alpha + \alpha - \alpha^2 \right)   \right] \\ & = &  \alpha +  \left( \left. m^{\rm mv} \right/ \left( m^{\rm Tot} \right) \right) \cdot \left( \alpha - \alpha^2 \right)  \\   {\rm Var} \left[ T_1 (\alpha) \right] & \longrightarrow & \left( m^{\rm Tot} \right)^{-2} \left \{ \sum_{j=1}^{m^{\rm st}} {\rm Var} \left[ {\mathbbm 1} \left( V_j \geq Q_V (1 - \alpha) \right) \right] + \sum_{j=1}^{m^{\rm mv}} {\rm Var} \left[ {\mathbbm 1} \left( Y_j \geq Q_Y (1 - \alpha) ~{\rm or}~Z_j \leq Q_Z \left( \alpha \right)  \right) \right] \right \} \\ & = & \left( m^{\rm Tot} \right)^{-2} \left \{    m^{\rm st} \alpha (1 - \alpha) + m^{\rm mv} {\rm Var} \left[ {\mathbbm 1} \left( Y_j \geq Q_Y \left( 1 - \alpha \right) \right)   \right] + m^{\rm mv} {\rm Var} \left[ {\mathbbm 1} \left( Z_j \leq Q_Z(\alpha) \right) \right]  \right.  \\ & & \left.  + m^{\rm mv} {\rm Var} \left[  {\mathbbm 1} \left( Y_j \geq Q_Y(1-\alpha) ~{\rm and}~Z_j \leq Q_Z(\alpha) \right) \right]  \right \}   \\ & = & \left( m^{\rm Tot} \right)^{-2} \left[   m^{\rm st} \cdot \alpha \cdot (1 - \alpha) + 2 \cdot m^{\rm mv} \cdot \alpha \cdot (1-\alpha) + m^{\rm mv} \alpha^2 \cdot (1 - \alpha^2)  \right]  \end{eqnarray*}    \end{proof}

\begin{thm} \label{thm:BottsDistn}  With the same definitions of ${\mathcal V}^{\mathcal D},$ ${\mathcal V}^{\mathcal N}$, etc. established in Theorem \ref{thm:LiuExpectation}, consider the statistic $T_2$, where: \footnote{In this theorem $T_2$ takes the form of ${\rm Anom}_{\rm Botts} \left( {\mathcal N} \right).$}     $$  T_2  =    \left \{  \begin{array}{ll}  S_1  & m^{\rm mv} = 0 \\ S_2 & m^{\rm st} = 0 \\ \left( S_1 + S_2 \right)/\sqrt{2} & m^{\rm st} > 0~\& ~m^{\rm mv} > 0  \end{array}  \right.    ,   $$        \begin{eqnarray*}  S_1 & = & \left. \left \{ \left(  \left( m^{\rm st} \right)^{-1} \sum_{j=1}^{m^{\rm st}} \left( 1 - {\hat F}_{V,r^{\rm st}}(V_j^{\mathcal N}) \right) -  {\frac{1}{2}} \right) \right \} \right/ \sqrt{{\frac{1}{12 m^{\rm st}}}}  , \\ S_2 & = & \left. \left \{ \left[     \left( m^{\rm mv} \right)^{-1} \sum_{j=1}^{m^{\rm mv}} \min \left( \left( 1 - {\hat F}_{Y, r^{\rm mv}}(Y_j^{\mathcal N}) \right), {\hat F}_{Z, r^{\rm mv}}(Z_j^{\mathcal N}) \right)    \right] - {\frac{1}{3}} \right \} \right/ \sqrt{ {\frac{1}{18 m^{\rm mv}}}},   \end{eqnarray*} and ${\hat F}_{V, r^{\rm st}} \left( V_j^{\mathcal N} \right) = \left( r^{\rm st} \right)^{-1} \sum_{i=1}^{r^{\rm st}} {\mathbbm 1} \left( V_i^{\mathcal D} \leq V_j^{\mathcal N} \right).$    Then  $$ T_2   \stackrel{d}{\longrightarrow}  N(0,1) ~~{\rm as}~ {\tilde {m}}^{\mathcal N} ~\&~ \min \left( r^{\rm st}, r^{\rm mv} \right)  \longrightarrow \infty,$$ where   $${\tilde {m}}^{\mathcal N} = \left \{ \begin{array}{ll}  m^{\rm st} & {\rm if}~m^{\rm mv} = 0 \\ m^{\rm mv} & {\rm if}~m^{\rm st} = 0 \\   \min \left( m^{\rm st}, m^{\rm mv} \right) & {\rm if} ~m^{\rm st} > 0~\&~m^{\rm mv} > 0    \end{array} \right. .$$    \end{thm}

\begin{proof} For calculating the asymptotic distribution of $T_2$, we first have to remember that from the Dvoretsky-Kiefer-Wolfowitz  inequality (\cite{DKW}), we get \begin{eqnarray}   \nonumber  {\mathbbm P} \left \{  \sup_{v \in {\mathbbm R}} \left|  {\hat F}_{V, r^{\rm st}} \left( v \right) - F_V(v)   \right| \geq \epsilon   \right \} & \leq & 2 \exp \left \{ -2 r^{\rm st} \epsilon^2  \right \}    \\ \label{eqn:ConvProb} \Longrightarrow   {\mathbbm P} \left \{   \left|  {\hat F}_{V, r^{\rm st}} \left( V_j^{\mathcal N} \right) - F_V \left( V_j^{\mathcal N} \right)   \right| \geq \epsilon   \right \} & \leq & 2 \exp \left \{ -2 r^{\rm st} \epsilon^2  \right \}.    \end{eqnarray} From Equation \ref{eqn:ConvProb} it follows that ${\hat F}_{V, r^{\rm st}} \left( V_j^{\mathcal N} \right) \stackrel{p}{\longrightarrow} F_V \left( V_j^{\mathcal N} \right)$ as $r^{\rm st} \longrightarrow \infty.$ Since convergence in probability implies convergence in distribution, we get $${\hat F}_{V, r^{\rm st}} \left( V_j^{\mathcal N} \right) \stackrel{d}{\longrightarrow} F_V \left( V_j^{\mathcal N} \right)~{\rm as}~r^{\rm st} \longrightarrow \infty.$$ The quantities $S_1$ and $S_2$ thus, respectively, converge in distribution to the values $S_1^*$ and $S_2^*$ as $r^{\rm st}, r^{\rm mv} \longrightarrow \infty,$ where  \begin{eqnarray*}  S_1^* & = & \left. \left \{ \left[  \left( m^{\rm st} \right)^{-1} \sum_{j=1}^{m^{\rm st}} \left( 1 - F_{V}(V_j^{\mathcal N}) \right) \right] -  {\frac{1}{2}}  \right \} \right/ \sqrt{{\frac{1}{12 m^{\rm st}}}}  ,~~{\rm and} \\ S_2^* & = & \left. \left \{ \left[     \left( m^{\rm mv} \right)^{-1} \sum_{j=1}^{m^{\rm mv}} \min \left( \left( 1 - F_{Y,}(Y_j^{\mathcal N}) \right), F_{Z}(Z_j^{\mathcal N}) \right)    \right] - {\frac{1}{3}} \right \} \right/ \sqrt{ {\frac{1}{18 m^{\rm mv}}}}   \end{eqnarray*}

  For calculating the asymptotic distributions of $S_1^*$ and $S_2^*$ as $m^{\rm st}$ \& $m^{\rm mv} \longrightarrow \infty$, it is important to remember that $1 - F_V \left( V_j^{\mathcal N} \right) \sim {\rm Unif}(0,1)$. Since the expected value and variance of a uniform random variable are ${\frac{1}{2}}$ and ${\frac{1}{12}}$, respectively, it follows (by the central limit theorem) that $S_1^* \stackrel{d}{\longrightarrow} N(0,1)$ as $m^{\rm st} \longrightarrow \infty.$ 

The random variables $\left( 1 - F_Y \left(Y_j^{\mathcal N} \right) \right)$ and $F_Z \left(Z_j^{\mathcal N} \right)$ also follow uniform distributions, making $A_j = \min \left[ \left(1 - F_Y \left( Y_j^{\mathcal N} \right) \right), F_Z \left( Z_j^{\mathcal N} \right) \right] \sim f_A(a),$ where $$f_A(a) = 2 (1 - a)~~~0 \leq a \leq 1.$$    The expected value and variance of $A$ are ${\frac{1}{3}}$ and ${\frac{1}{18}}$, respectively, implying (by the central limit theorem) that $S_2^* \stackrel{d}{\longrightarrow} N(0,1)$ as $m^{\rm mv} \longrightarrow \infty$.

With $S_1 \stackrel{d}{\longrightarrow} N(0,1)$ as $r^{\rm st}~\&~m^{\rm st} \longrightarrow \infty$,  $S_2 \stackrel{d}{\longrightarrow} N(0,1)$ as $r^{\rm mv}~\&~m^{\rm mv} \longrightarrow \infty$, it follows that $ \left. \left( S_1 + S_2 \right) \right/ \sqrt{2} \stackrel{d}{\longrightarrow} N(0,1)$  as ${\tilde m}^{\mathcal N}$ \& $\min(r^{\rm st},  r^{\rm mv}) \longrightarrow \infty$, making $T_2 \stackrel{d}{\longrightarrow} N(0,1)$ as ${\tilde m}^{\mathcal N}$ \& $\min(r^{\rm st},  r^{\rm mv}) \longrightarrow \infty$.  \end{proof}

\section{Simulations}\label{sctn:Simulations}

In this section, we perform simulations which illustrate the asymptotic result proven in Theorem \ref{thm:BottsDistn}. The histogram and Q-Q plot of 5000 values $T_2$ are provided in Figures \ref{fig:CLT_Demo1} - \ref{fig:CLT_Demo4}. The conditions of each simulation are provided in the caption.

{
\begin{figure}[H] 
\begin{center}
{\includegraphics[scale = 0.5]{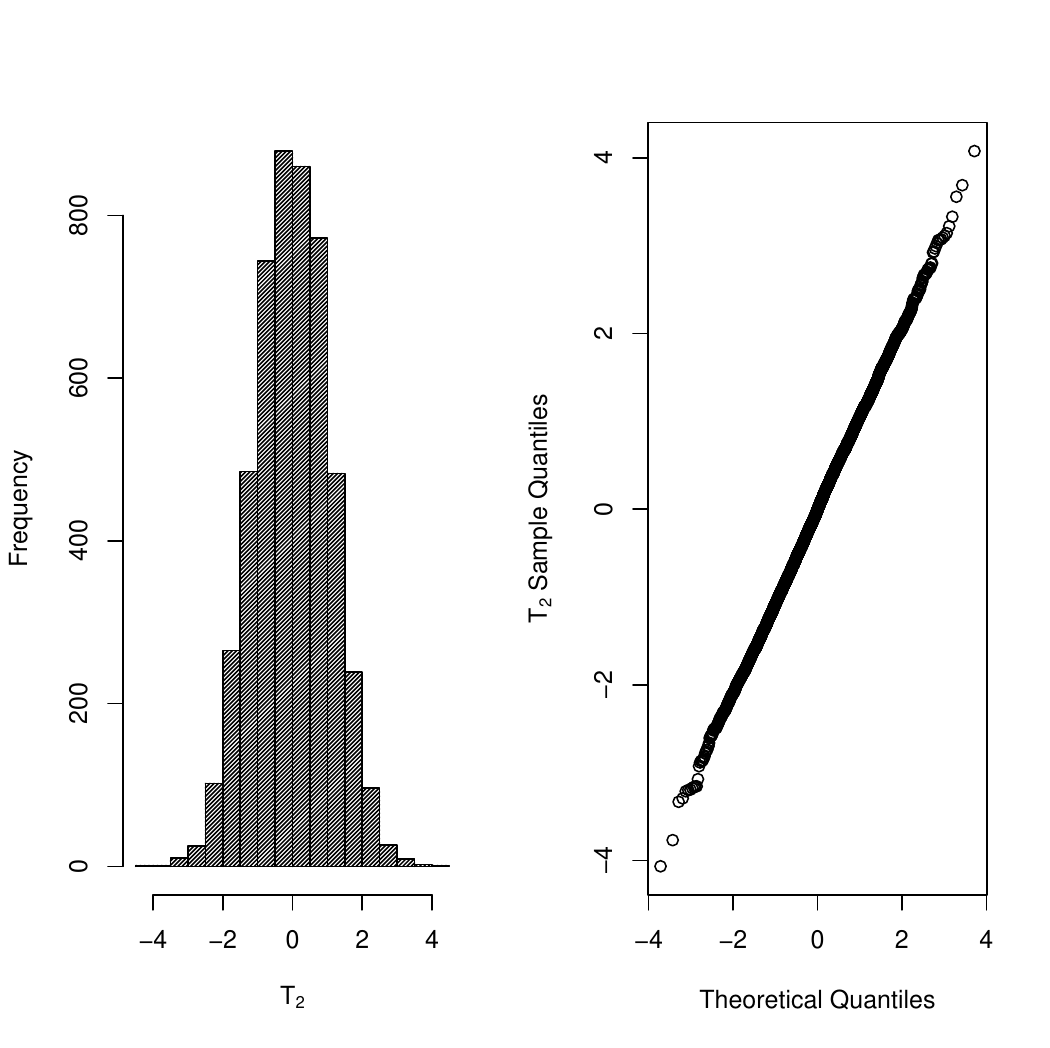}}

\caption{Histogram and Q-Q plot of 5000 values of $T_2.$   In this case, $r^{\rm st} = r^{\rm mv} = 1000$, $f_V(v)$ is an exponential distribution with $\lambda = 2$, $f_Y(y)$ is a gamma distribution with $\alpha = 2$ and $\beta = 4$, and $f_Z(z)$ is a chi-squared distribution with 8 degrees of freedom.   For this simulation, $m^{\rm st} = 100$ and $m^{\rm mv} = 200$. }

\label{fig:CLT_Demo1}
\end{center}
\end{figure}
}

{
\begin{figure}[H] 
\begin{center}
{\includegraphics[scale = 0.5]{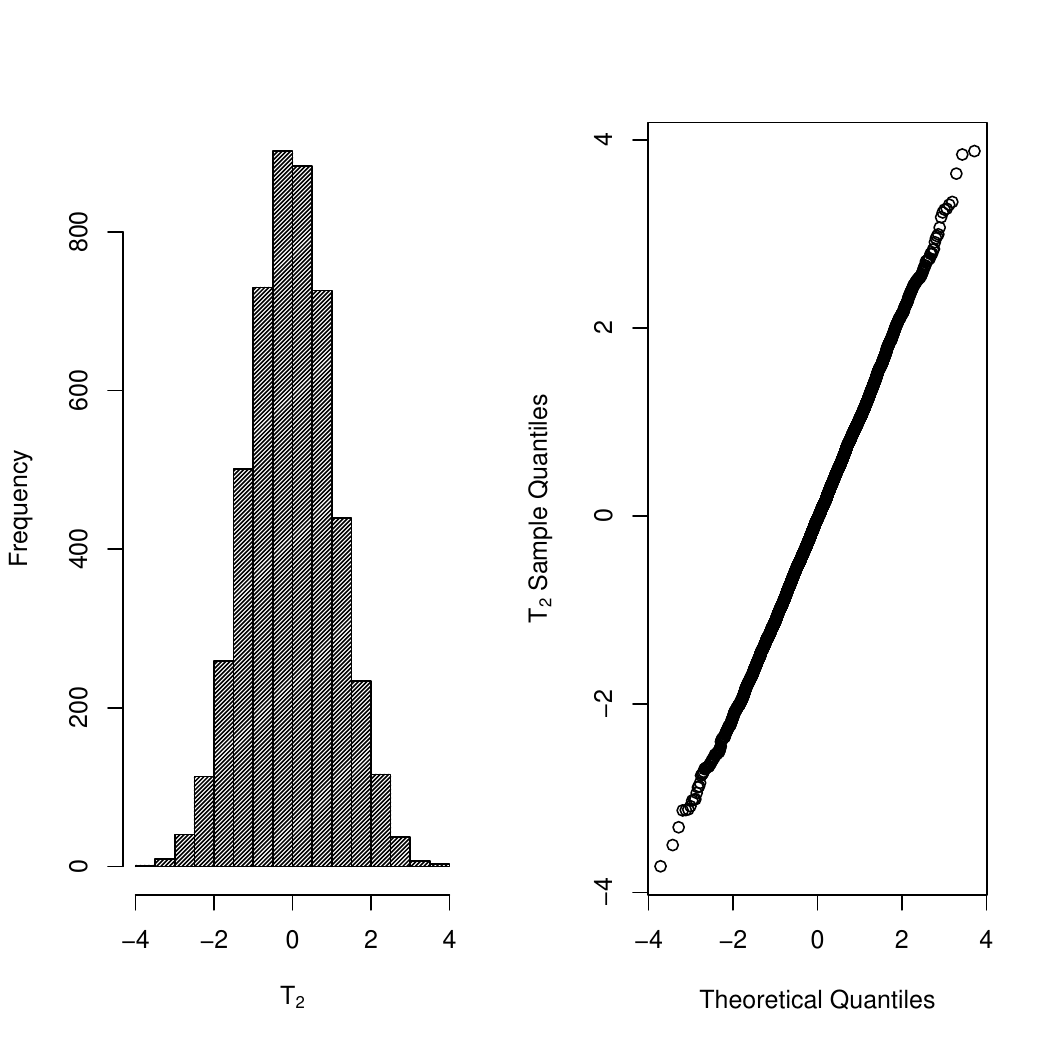}}

\caption{Histogram and Q-Q plot of 5000 values of $T_2.$   In this case, $r^{\rm st} = r^{\rm mv} = 1000$, $f_V(v)$ is a normal distribution with $\mu = 2$, and $\sigma = 4$, $f_Y(y)$ is a cauchy   distribution with $y_0 = 0$ and $\gamma = 1$, and  $f_Z(z)$ is an F distribution with  $d_1 = 8$ and $d_2 = 18$.   For this simulation, $m^{\rm st} = 200$ and $m^{\rm mv} = 100$. }

\label{fig:CLT_Demo2}
\end{center}
\end{figure}
}

{
\begin{figure}[H] 
\begin{center}
{\includegraphics[scale = 0.5]{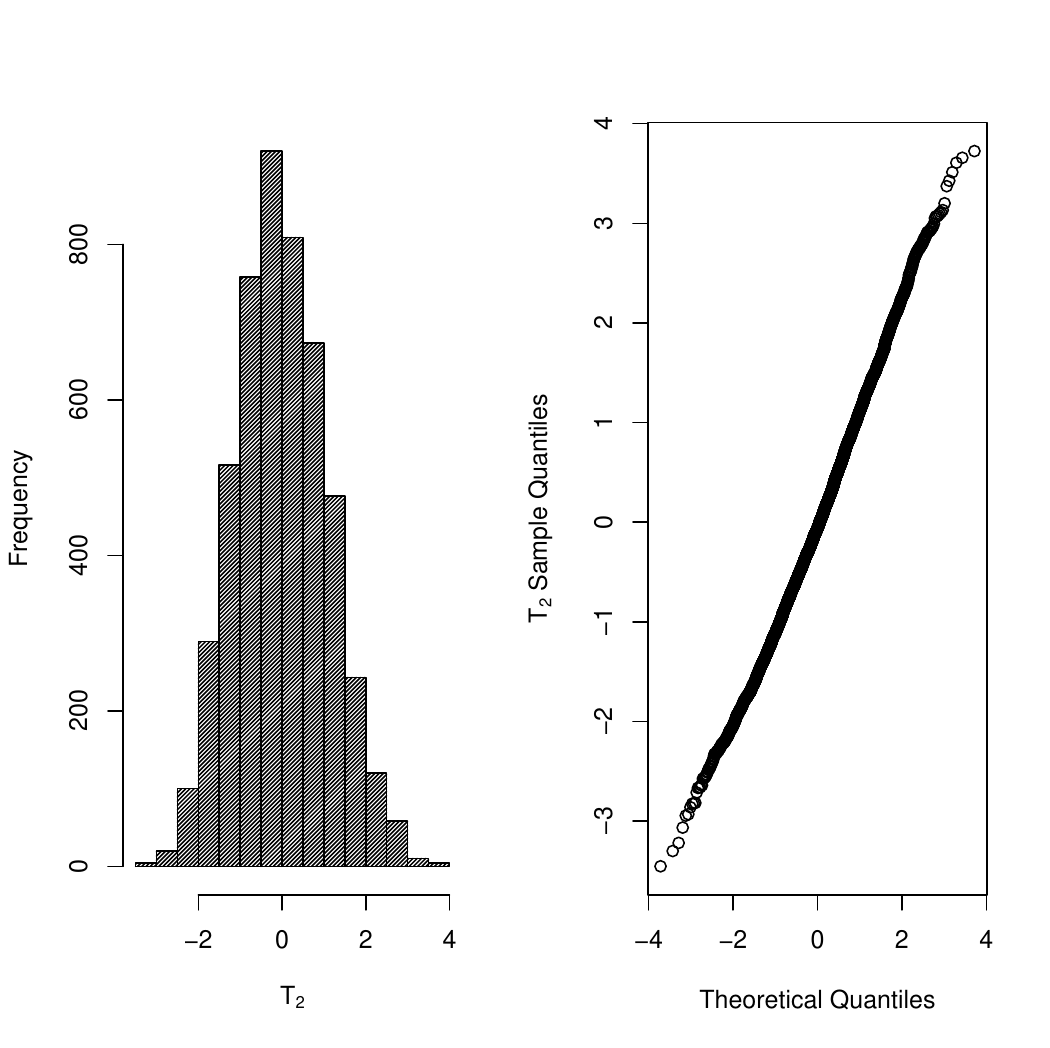}}

\caption{Histogram and Q-Q plot of 5000 values of $T_2.$   In this case, $r^{\rm st} = 1000$, $r^{\rm mv} = 1000$, and  $f_V(v)$ is a t distribution with $\nu = 2$.   For this simulation, $m^{\rm st} = 300$, and $m^{\rm mv} = 0.$ }

\label{fig:CLT_Demo3}
\end{center}
\end{figure}
}

{
\begin{figure}[H] 
\begin{center}
{\includegraphics[scale = 0.5]{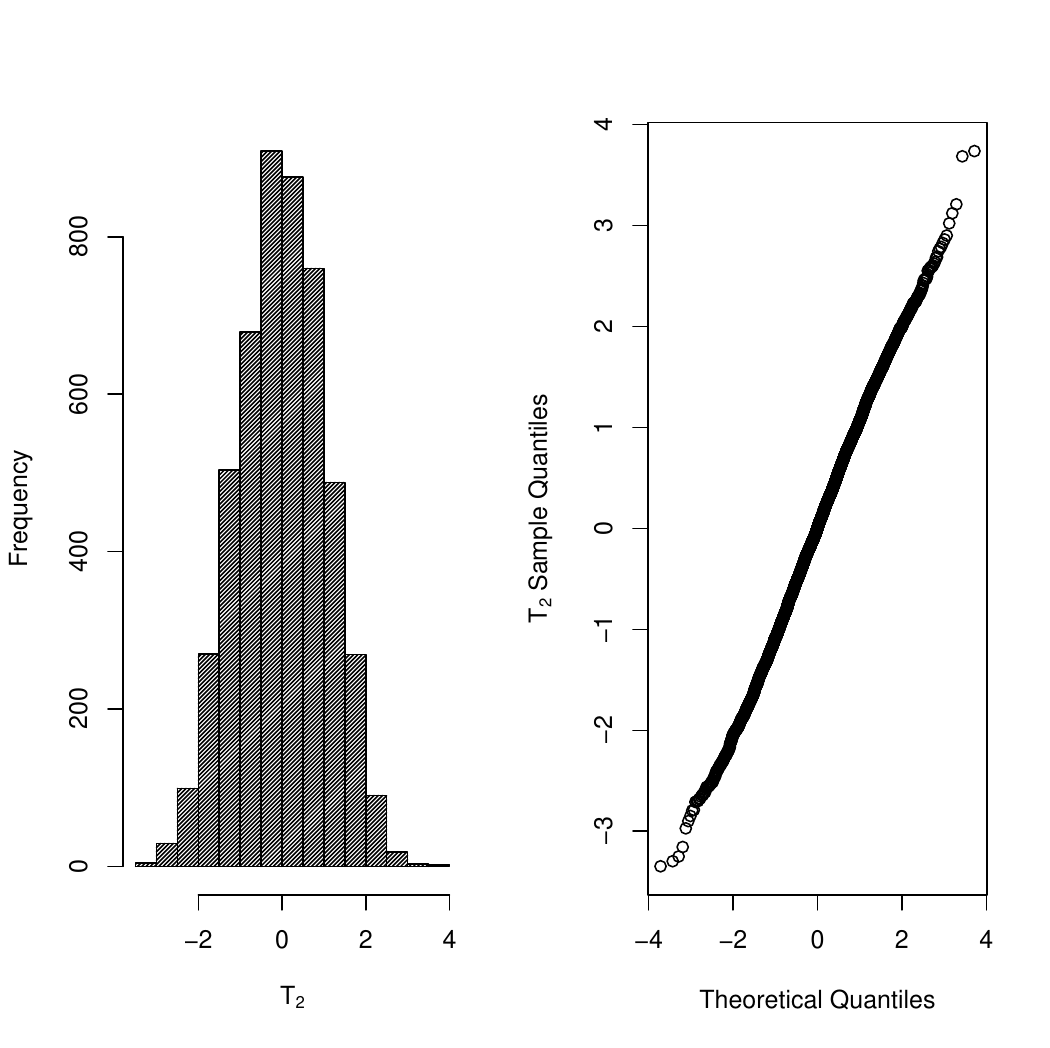}}

\caption{Histogram and Q-Q plot of 5000 values of $T_2.$   In this case, $r^{\rm st} = 1000$, $r^{\rm mv} = 1000$, $f_Y(y)$ is a cauchy   distribution with $y_0 = 4$ and $\gamma = 2$, and  $f_Z(z)$ is an F distribution with  $d_1 = 1$ and $d_2 = 20$.   For this simulation, $m^{\rm mv} = 200,$ and $m^{\rm st} = 0$. }

\label{fig:CLT_Demo4}
\end{center}
\end{figure}
}

\end{document}